\documentclass[twoside,11pt]{article}
\usepackage{jmlr2e}
\usepackage{bm}
\usepackage{graphicx,subfigure}
\usepackage{float}
\usepackage{boldline}
\usepackage{array}
\usepackage{amsmath}
\usepackage{thmtools}
\usepackage[nottoc]{tocbibind}
\usepackage[utf8]{inputenc}
\usepackage{placeins}
\usepackage{tikz}
\usetikzlibrary{calc}
\usepackage{hyperref}
\usepackage{chngcntr,apptools}
 
\newcommand{\bcdot}{\boldsymbol{\cdot}}

\renewcommand{\cite}{\citep}
\renewcommand{\vec}{\bm}

\DeclareMathOperator*{\var}{var}

\AtAppendix{\counterwithin{theorem}{section}}
\AtAppendix{\counterwithin{table}{section}}
\AtAppendix{\counterwithin{equation}{section}}

\ShortHeadings{Deep Portfolio Management}{Jiang, Xu and Liang}
\firstpageno{1}

\begin{document}

\title{A Deep Reinforcement Learning Framework for the Financial Portfolio Management Problem}

\author{\name{Zhengyao Jiang} \email{zhengyao.jiang15@student.xjtlu.edu.cn}\\
	\name{Dixing Xu} \email{dixing.xu15@student.xjtlu.edu.cn}\\
	\addr Department of Computer Sciences and Software Engineering\\
	\name{Jinjun Liang} \email{jinjun.liang@xjtlu.edu.cn} \\
	\addr Department of Mathematical Sciences\\
	Xi'an Jiaotong-Liverpool University \\
	Suzhou, SU 215123, P.\ R.\ China
	}

\editor{XZY ABCDE}

\maketitle

\begin{abstract}
Financial portfolio management is the process of constant redistribution of a fund into different financial products. 
This paper presents a financial-model-free Reinforcement Learning framework to provide a deep machine learning solution to the portfolio management problem.
The framework consists of the Ensemble of Identical Independent Evaluators (EIIE) topology, a Portfolio-Vector Memory (PVM), an Online Stochastic Batch Learning
(OSBL) scheme, and a fully exploiting and explicit reward function.
This framework is realized in three instants in this work with a Convolutional Neural Network (CNN), a basic Recurrent Neural Network (RNN), 
and a Long Short-Term Memory (LSTM).
They are, along with a number of recently reviewed or published portfolio-selection strategies, examined
in three back-test experiments with a trading period of 30 minutes in a cryptocurrency market.
Cryptocurrencies are electronic and decentralized
alternatives to government-issued money, with Bitcoin as the
best-known example of a cryptocurrency.
All three instances of the framework monopolize the top three positions in all experiments, outdistancing other compared trading algorithms.
Although with a high commission rate of $0.25\%$ in the back-tests, the framework is able to achieve at least 4-fold returns in 50 days.

\end{abstract}


\section{Introduction}\label{introduction}


Portfolio management is the decision making process of continuously reallocating an amount of fund into a number of different financial investment products,
aiming to maximize the return while restraining the risk \cite{haugen1986inv,markowitz1968portfolio}. 
Traditional portfolio management methods can be classified into four categories, "Follow-the-Winner", "Follow-the-Loser",
"Pattern-Matching", and "Meta-Learning" \cite{li2014survey}.
The first two categories are based on prior-constructed financial models, while they may also be assisted by some machine learning techniques 
for parameter determinations \cite{li2012pamr,cover1996universal}. The performance of these methods is dependent on the validity of the models
on different markets.
"Pattern-Matching" algorithms predict the next market distribution based on a sample of historical data and
explicitly optimizes the portfolio based on the sampled distribution \cite{gyorfi2006nonparametric}.
The last class, "Meta-Learning" method combine multiple strategies of other categories to attain more consistent performance \cite{vovk1998universal,Das2011}.

There are existing deep machine-learning approaches to financial market trading. However, many of them try to predict price movements or
trends \cite{Heaton2016,Niaki2013,freitas2009prediction}.
With history prices of all assets as its input, a neural network can output a predicted vector of asset prices for the next period.
Then the trading agent can act upon this prediction. This idea is straightforward to implement, because it is a supervised learning,
or more specifically a regression problem.  The performance of these price-prediction-based algorithms, however, highly depends on the degree of
prediction accuracy, but it turns out that future market prices are difficult to predict. Furthermore, price predictions are not market actions, 
converting them into actions requires additional layer of logic. If this layer is a hand-coded, then the whole approach is not
fully machine learning, and thus is not very extensible or adaptable. For example, it is difficult for a prediction-based network to consider 
transaction cost as a risk factor.

Previous successful attempts of model-free and fully machine-learning schemes to the algorithmic trading problem, without predicting future prices, are treating 
the problem as a Reinforcement Learning (RL) one.  These include \citet{moody2001drl}, \citet{adempster2006daptive}, \citet{cumming2013masterthesis}, 
and the recent deep RL utilization by \citet{deng2017deep}.
These RL algorithms output discrete trading signals on an asset. Being limited to single-asset trading, they are not applicable to
general portfolio management problems, where trading agents manage multiple assets.

Deep RL is lately drawing much attention due to its remarkable achievements in playing video games \cite{Mnih2015} and board games \cite{silver2016alphago}. These
are RL problems with discrete action spaces, and can not be directly applied to portfolio selection problems, where actions are continuous.
Although market actions can be discretized, discretization is considered a major drawback, because discrete actions come
with unknown risks. For instance, one extreme discrete action may be defined as investing all the capital into one asset, without spreading the risk to
the rest of the market. In addition, discretization scales badly. Market factors, like number of total assets, vary from market to market.
In order to take full advantage of adaptability of machine learning over different markets, trading algorithms have to be scalable.
A general-purpose continuous deep RL framework, the actor-critic Deterministic Policy Gradient Algorithms, was recently introduced \cite{Silver2014,Lillicrap2016}.
The continuous output in these actor-critic algorithms is achieved by a neural-network approximated action policy function,
and a second network is trained as the reward function estimator.
Training two neural networks, however, is found out to be difficult, and sometimes even unstable.

This paper proposes an RL framework specially designed for the task of portfolio management. The core of the framework is the Ensemble of Identical
Independent Evaluators (EIIE) topology. An IIE is a neural network whose job is to inspect the history of an asset and evaluate its potential growth for 
the immediate future. The evaluation score of each asset is discounted by the size of its intentional weight change for the asset
in the portfolio and is presented to a softmax layer, whose outcome will be the new portfolio weights
for the coming trading period. The portfolio weights define the market action of the RL agent. An asset with an increased target weight will be bought in with
additional amount, and that with decreased weight will be sold. Apart from the market history, portfolio weights from the previous trading period
are also input to the EIIE. This is for the RL agent to consider the effect of transaction cost to its wealth. For this purpose, the portfolio weights of
each period are recorded in a Portfolio Vector Memory (PVM). The EIIE is trained in an Online Stochastic Batch Learning scheme (OSBL), which is compatible with
both pre-trade training and online training during back-tests or online trading. The reward function of the RL framework is the explicit average of the
periodic logarithmic returns. Having an explicit reward function, the EIIE evolves, under training, along the gradient ascending direction of the function.
Three different species of IIEs are tested in this work, a Convolutional Neural Network (CNN)
\cite{fukushima1980neocognitron,krizhevsky2012imagenet,sermanet2012convolutional}, 
a basic Recurrent Neural Network (RNN) \cite{werbos1988generalization},
and a Long Short Term Memory (LSTM) \cite{hochreiter1997long}.

Being a fully machine-learning approach,  the framework is not restricted to any particular markets. 
To examine its validity and profitability, the framework is tested in a cryptocurrency
(virtual money, Bitcoin as the most famous example) exchange market, Polonix.com.
A set of coins are preselected by their ranking in trading-volume over a time interval just before an experiment.
Three back-test experiments of well separated time-spans are performed in a trading period of 30 minutes.
The performance of the three EIIEs are compared with some recently published or reviewed portfolio selection strategies \cite{li2015olps,li2014survey}.
The EIIEs significantly beat all other strategies in all three experiments

Cryptographic currencies, or simply cryptocurrencies, are electronic and decentralized alternatives to government-issued moneys
\cite{nakamoto2008bitcoin,grinberg2012bitcoin}. While the best known example of a cryptocurrency is Bitcoin, there are more than 100 other
tradable cryptocurrencies competing each other and with Bitcoin \cite{bonneau2015sok}.
The motive behind this competition is that there are a number of design flaws in Bitcoin, and people are trying to invent new coins to overcome
these defects hoping their inventions will eventually replace Bitcoin \cite{bentov2014proof,duffield2014darkcoin}.
There are, however, more and more cryptocurrencies being created without targeting to beat Bitcoin, but with the purposes
of using the blockchain technology behind it to develop 
decentralized applications\footnote{For example, Ethereum is a decentralized platform that runs smart contracts, and Siacoin is the currency for 
buying and selling storage service on the decentralized cloud Sia.
}.
To June 2017, the total market capital of all cryptocurrencies is 102 billions in USD, 41 of which is of
Bitcoin.\footnote{Crypto-currency market capitalizations, http://coinmarketcap.com/, accessed: 2017-06-30.} Therefore, regardless of its design faults,
Bitcoin is still the dominant cryptocurrency in markets. As a result, many other currencies can not be bought with fiat currencies, but only be traded against Bitcoin.

Two natures of cryptocurrencies differentiate them from traditional financial assets, making their market the best test-ground for
algorithmic portfolio management experiments. These natures are decentralization and openness, and the former implies the latter.
Without a central regulating party, anyone can participate in cryptocurrency trading with low entrance requirements.
One direct consequence is abundance of small-volume currencies. Affecting the prices of these penny-markets will require smaller amount of investment,
compared to traditional markets. This will eventually allow trading machines to learn and take advantage of the impacts by their own market actions.
Openness also means the markets are more accessible. Most cryptocurrency exchanges have application programming interface for obtaining market data and
carrying out trading actions, and most exchanges are open 24/7 without restricting frequency of tradings. These non-stop markets are ideal for machines
to learn in the real world in shorter time-frames.

The paper is organized as follows. Section~\ref{section:problem} defines the portfolio management problem that this project is aiming to solve.
Section~\ref{section:data} introduces asset preselection and the reasoning behind it, the input price tensor, and a way to deal with missing data in the 
market history.
The portfolio management problem is re-described in the language RL in Section~\ref{section:rl}.
Section~\ref{section:networks} presents the EIIE meta topology, the PVM, the OSBL scheme.
The results of the three experiments are staged in Section~\ref{section:experiments}.

\section{Problem Definition}\label{section:problem}
Portfolio management is the action of continuous reallocation of a capital into a number of financial assets. For an automatic trading robot,
these investment decisions and actions are made periodically. This section provides a mathematical setting of the portfolio management problem.

\subsection{Trading Period}\label{section:period}
In this work, trading algorithms are time-driven, where time is divided into periods of equal lengths $T$.
At the beginning of each period, the trading agent reallocates the fund among the assets. 
$T=30\,\mathrm{minutes}$  in all experiments of this paper. The price of an asset goes up and down within a period,
but four important price points characterize the overall movement of a period,
namely the opening, highest, lowest and closing prices \cite{rogers1991estimating}.
For continuous markets, the opening price of a financial
instrument in a period is the closing price from the previous period. It is assumed in the back-test experiments that 
at the beginning of each period assets can be bought or sold at the opening price of that period. The justification of 
such an assumption is given in Section~\ref{section:hypotheses}.

\subsection{Mathematical Formalism}\label{section:math}
The portfolio consists of $m$ assets. The closing prices of all assets comprise
the \emph{price vector} for Period $t$, $\vec{v}_t$. In other words, the $i\mathrm{th}$ element of $\vec{v}_t$, $v_{i,t}$, is the closing price
of the $i\mathrm{th}$ asset in the $t$th period.
Similarly, $\vec{v}^{(\mathrm{hi})}_t$ and $\vec{v}^{(\mathrm{lo})}_t$
denote the highest and lowest prices of the period.
The first asset in the portfolio is special, that it is the quoted currency, referred to as \emph{the cash} for the
rest of the article. Since the prices of all assets are quoted in cash, the first elements of $\vec{v}_t$, 
$\vec{v}^{(\mathrm{hi})}_t$ and $\vec{v}^{(\mathrm{lo})}_t$ are  
always one, that is $v_{0,t}^{(\mathrm{hi})}=v_{0,t}^{(\mathrm{lo})}=v_{0,t}= 1,\,\forall t$. In the experiments of 
this paper, the cash is Bitcoin. 

For continuous markets, elements of $\vec{v}_{t}$ are the opening prices for Period $t+1$ as well as the closing prices for Period $t$.
The \emph{price relative vector} of the $t$th trading period, $\vec y_{t}$, is defined as the element-wise division of $\vec v_{t}$ by $\vec v_{t-1}$:
\begin{equation}
	\vec{y}_{t} := \vec{v}_{t} \oslash \vec{v}_{t-1} =
\left( 1, \frac{ v_{1,t}}{v_{1,t-1}}, \frac{v_{2,t}}{v_{2,t-1}}, ..., \frac{v_{m,t}}{v_{m,t-1}} \right)^\intercal.
\label{eq:y}
\end{equation}
The elements of $\vec{y}_{t}$ are the quotients of closing prices and opening prices for individual asset in the period. 
The price relative vector can be used to calculate the change in total portfolio value in a period. 
If $p_{t-1}$ is the portfolio value at the begining of Period $t$, ignoring transaction cost,
\begin{equation} \label{eq:portfolio_value_t_no_cost}
	p_{t} = p_{t-1} \, \vec{y}_{t} \bcdot \vec{w}_{t-1},
\end{equation}
where $\vec{w}_{t-1}$ is the portfolio weight vector (referred to as the \emph{portfolio vector} from now on) at the beginning of Period $t$, whose $i$th
element, $w_{t-1,i}$,  is the proportion of asset $i$ in the portfolio after capital reallocation. 
The elements of $\vec{w}_t$ always sum up to one by definition, $\sum\limits_i  w_{t,i} =1, \forall t$. The \emph{rate of return} for Period $t$ is then
\begin{equation} \label{eq:rho_t_no_cost}
	\rho_{t} := \frac{p_{t}}{p_{t-1}} -1  = \vec{y}_{t} \bcdot \vec{w}_{t-1} -1,
\end{equation}
and the corresponding \emph{logarithmic rate of return} is
\begin{equation} 
	r_{t} := \ln \frac{p_{t}}{p_{t-1}} = \ln \vec{y}_{t} \bcdot \vec{w}_{t-1}.
	\label{eq:log_r_no_cost}
\end{equation}

In a typical portfolio management problem, the initial portfolio weight vector $\vec{ w}_0$ is chosen to be the first basis vector in the Euclidean space,
\begin{equation}
	\vec  w_0=(1, 0, ...,0)^\intercal,
	\label{eq:w_0}
\end{equation}
indicating all the capital is in the trading currency before entering the market.
If there is no transaction cost, the final portfolio value will be
\begin{equation}\label{eq:final_P_no_cost}
	p_\mathrm{f} = p_0 \exp\left( \sum\limits_{t=1}^{t_\mathrm{f}+1} r_{t} \right) = p_0 \prod_{t=1}^{t_\mathrm{f}+1} \vec{y}_{t} \bcdot \vec{w}_{t-1},
\end{equation}
where $p_0$ is the initial investment amount. The job of a portfolio manager is to maximize $p_\mathrm{f}$ for 
a given time frame. 

\subsection{Transaction Cost}\label{section:transaction}

In a real-world scenario, buying or selling assets in a market is not free. The cost is normally from commission fee.
Assuming a constant commission rate, this section will re-calculate the final portfolio value in Equation~\eqref{eq:final_P_no_cost}, 
using a recursive formula extending a work by \citet{ormos2013}.

The portfolio vector at the beginning of Period $t$ is $\vec{w}_{t-1}$. Due to price movements in the market, at the end of the same
period, the weights evolve into 
\begin{equation} \label{eq:w_prime}
   \vec{w}'_{t} = \frac{ \vec{y}_{t} \odot \vec{w}_{t-1} }
			{ \vec{y}_{t} \bcdot \vec{w}_{t-1} },
\end{equation}
where $\odot$ is the element-wise multiplication.
The mission of the portfolio manager now at the end of Period $t$ is to reallocate portfolio vector from $\vec{w}'_{t}$
to $\vec{w}_{t}$ by selling and buying relevant assets. Paying all commission fees, this reallocation action shrinks
the portfolio value by a factor $\mu_{t}$. $\mu_t \in (0,1]$, and will be called the \emph{transaction remainder factor} from now on.
$\mu_{t}$ is to be determined below.
Denoting $p_{t-1}$ as the portfolio value at the beginning of Period $t$ and $p'_t$ at the end,
\begin{equation} \label{eq:mu_def}
	p_t = \mu_t p'_t. 
\end{equation}
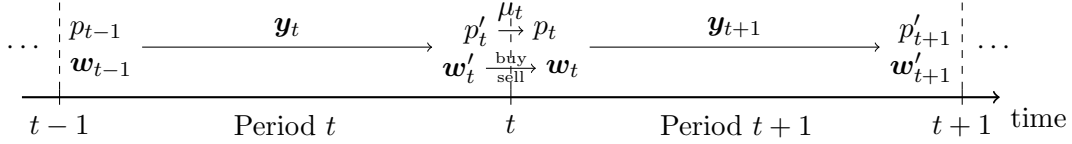
\begin{figure}[!ht]
\centering
\begin{tikzpicture} [scale=1]
	\def\len{13}
	\def\hight{8ex}
	\def\tic{0.1}
	\def\wy{1em}
	\def\py{2.3em}
	\def\offset{0.5}
	\def\T{(\len -2*\offset)/2}
	\def\Pone{{\offset+\T/2}}
	\def\Ptwo{{\offset+3*\T/2}}
	\def\Tone{\offset}
	\def\Ttwo{{\offset+\T}}
	\def\Tthr{{\offset+2*\T}}
	\def\yshift{1.7ex}
	\draw [thick, ->] (0,0) -> (\len,0) node [below right] {time};
	\draw (\Tone,\tic) -- (\Tone,-\tic) node [below] {$t-1$};
	\draw (\Ttwo,\tic) -- (\Ttwo,-\tic) node [below] {$t$};
	\draw (\Tthr,\tic) -- (\Tthr,-\tic) node [below] {$t+1$};
	\node [below] at (\Pone,-\tic) {Period $t$};
	\node [below] at (\Ptwo,-\tic) {Period $t+1$};
	\draw [dashed] (\Ttwo,\tic) -- (\Ttwo,{\tic+\hight});
	\draw [dashed] (\Tone,\tic) -- (\Tone,{\tic+\hight});
	\draw [dashed] (\Tthr,\tic) -- (\Tthr,{\tic+\hight});
	\node [right] at (\Tone, \wy) (wz) {$\vec w_{t-1}$};
	\node [right=2ex] at (\Ttwo, \wy) (wb) {$\vec w_{t}$};
	\node [right] at (\Tone, \py) (pz) {$p_{t-1}$};
	\node [right=1ex] at (\Ttwo, \py) (pb) {$p_{t}$};
	\node [left] at (\Tthr, \wy) (wc) {$\vec w'_{t+1}$};
	\node [left] at (\Tthr, \py) {$p'_{t+1}$};
	\node [left=2ex] at (\Ttwo, \wy) (wa) {$\vec w'_{t}$};
	\node [left=1ex] at (\Ttwo, \py) (pa) {$p'_{t}$};
	\draw [->] (pa) -- node [above] {$\mu_t$} ++ (pb);
	\draw [->] (wa) -- node [above=-3pt] {\tiny buy}node [below=-3pt] {\tiny sell} ++ (wb);
	\draw [->, transform canvas={yshift=\yshift}] (wz) -- node [above] {$\vec y_{t}$} ++ (wa);
	\draw [->, transform canvas={yshift=\yshift}] (wb) -- node [above] {$\vec y_{t+1}$} ++ (wc);
	\node [left of= wz, transform canvas={yshift=\yshift}] {$\cdots$};
	\node [right of= wc, transform canvas={yshift=\yshift}] {$\cdots$};
\end{tikzpicture}
\caption{Illustration of the effect of transaction remainder factor $\mu_t$. The market movement during Period $t$, represented by the price-relative
vector $\vec y_t$, drives the portfolio value and portfolio weights from $p_{t-1}$ and $\vec w_{t-1}$ to $p'_t$ and $\vec w'_t$. The asset selling
and purchasing action at time $t$ redistributes the fund into $\vec w_t$. As a side-effect, these transactions shrink the portfolio to $p_t$
by a factor of $\mu_t$. The rate of return for Period $t$ is calculated using portfolio values at the beginning of the two consecutive periods in 
Equation~\eqref{eq:rho_t}.}
\label{fig:mu}
\end{figure}

The rate of return \eqref{eq:rho_t_no_cost} and logarithmic rate of return \eqref{eq:log_r_no_cost} are now
\begin{align}
	\rho_{t} &= \frac{p_{t}}{p_{t-1}} - 1 = \frac{ \mu_{t} p'_t  }{ p_{t-1} } -1 
	= \mu_t \vec y_{t} \bcdot \vec w_{t-1} -1, 
	\label{eq:rho_t} \\
	r_{t} &= \ln \frac{p_{t}}{p_{t-1}} 
	= \ln \left( \mu_t \vec y_{t} \bcdot \vec w_{t-1} \right),
	\label{eq:r_t}
\end{align}
and the final portfolio value in Equation~\eqref{eq:final_P_no_cost} becomes
\begin{equation}\label{eq:final_P}
	p_\mathrm{f} = p_0 \exp\left( \sum\limits_{t=1}^{t_\mathrm{f}+1} r_{t} \right)
	= p_0 \prod_{t=1}^{t_\mathrm{f}+1} \mu_{t} \vec{y}_{t} \bcdot \vec{w}_{t-1}.
\end{equation}
Different from Equation~\eqref{eq:log_r_no_cost} and~\eqref{eq:portfolio_value_t_no_cost} where transaction cost is not considered, in Equation~\eqref{eq:r_t}
and~\eqref{eq:final_P}, $p'_t \neq p_t$ and the difference between the two values is where the transaction remainder factor comes into play.
Figure~\ref{fig:mu} demonstrates the relationship among portfolio vectors and values and their dynamic relationship on a time axis.

The remaining problem is to determine this transaction remainder factor $\mu_{t}$. During the portfolio reallocation from $\vec{w}'_{t}$ to $\vec{w}_{t}$,
some or all amount of asset $i$ need to be sold, if $p'_t w'_{t,i} > p_t w_{t,i}$ or $w'_{t,i} > \mu_{t} w_{t,i}$.
The total amount of cash obtained by all selling is 
\begin{equation} \label{eq:selling}
    (1-c_\mathrm{s}) p'_t \sum_{i=1}^m (w'_{t,i} - \mu_{t} w_{t,i})^+ 
\end{equation}
where $0\leqslant c_\mathrm{s} <1$ is the commission rate for selling, and $(\vec{v})^+=\mathrm{ReLu}(\vec{v})$ is the element-wise rectified linear function,
$(x)^+ = x$ if $x>0$, $(x)^+ =0$ otherwise.
This money and the original cash reserve $p'_t w'_{t,0}$ taken away the new reserve  $\mu_{t} p'_t w_{t,0}$ will be used to buy new assets,
\begin{equation} \label{eq:purchasing}
   (1-c_\mathrm{p})
   \left[ w'_{t,0} + (1-c_\mathrm{s}) \sum_{i=1}^m (w'_{t,i} - \mu_{t} w_{t,i})^+ - \mu_{t} w_{t,0} \right]
   = \sum_{i=1}^m (\mu_{t} w_{t,i} - w'_{t,i})^+,
\end{equation}
where  $0\leqslant c_\mathrm{p} <1$ is the commission rate for purchasing, and $p'_t$ has been canceled out on both sides.
Using identity
$
   (a-b)^+ - (b-a)^+ = a-b,
$
and the fact that
$
    w'_{t,0} + \sum\limits_{i=1}^m w'_{t,i} = 1 = w_{t,0} + \sum\limits_{i=1}^m w_{t,i},
$
Equation~\eqref{eq:purchasing} is simplified to
\begin{equation} \label{eq:mu_iteration}
    \mu_{t} = 
    	    \frac{1}
	     { 1 - c_\mathrm{p} w_{t,0} }
	    \left[
	    1 - c_\mathrm{p} w'_{t,0} - (c_\mathrm{s}+c_\mathrm{p} - c_\mathrm{s}c_\mathrm{p})
	     \sum_{i=1}^m (w'_{t,i} - \mu_{t} w_{t,i})^+ 
     \right].
\end{equation}
The presence of $\mu_{t}$ inside a linear rectifier means $\mu_{t}$ is not solvable analytically, but it can only be solved
iteratively.
\begin{restatable}{thm}{convergence}
	\label{thm:convergence}
	Denoting 
	\[
	f(\mu) := 
    	    \frac{1}
	     { 1 - c_\mathrm{p} w_{t,0} }
	    \left[
	    1 - c_\mathrm{p} w'_{t,0} - (c_\mathrm{s}+c_\mathrm{p} - c_\mathrm{s}c_\mathrm{p})
	     \sum_{i=1}^m (w'_{t,i} - \mu w_{t,i})^+ 
        \right],
        \]
	the sequence $\left\{\tilde{\mu}^{(k)} _t\right\}$, defined as
	\begin{equation} \label{eq:sequence}
		\left\{\tilde{\mu}^{(k)} _t \left| \tilde{\mu}^{(0)} _t = \mu_\odot \;\mathrm{and}\; 
			\tilde{\mu}^{(k)} _t = f\left(\tilde{\mu}^{(k-1)} _t\right),\;
		k \in \mathbb{N}_0 \right. \right\}
	\end{equation}
	converges to $\mu_{t}$, the solution to Equation \eqref{eq:mu_iteration}, for any $\mu_\odot \in [0,1]$.
\end{restatable}

While this convergence is not stated in \citet{ormos2013}, its proof will be given in Appendix~\ref{section:proof}.
This theorem provides a way to approximate the transaction remainder factor $\mu_{t}$ to an arbitrary accuracy.
The speed on the convergence depends on the error of the initial guest $\mu_\odot$. The smaller $\left| \mu_t - \mu_\odot \right|$ is, the quicker 
Sequence~\eqref{eq:sequence} converges to $\mu_t$. When $c_\mathrm{p} = c_\mathrm{s} = c$, there is a practice \cite{moody1998performance} to approximate $\mu_t$ with
$c \sum\limits_{i=1}^m | w'_{t,i} - w_{t,i} |$. Therefore, in this work, $\mu_\odot$ will use this as the first value for the sequence, that
\begin{equation} \label{eq:mu_approx}
	\mu_\odot = c \sum\limits_{i=1}^m \left| w'_{t,i} - w_{t,i} \right|.
\end{equation}
In the training of the neural networks,
$\tilde{\mu}^{(k)} _t$ with a fixed $k$ in \eqref{eq:sequence} is used. In the back-test experiments, a tolerant error $\delta$ 
dynamically determines $k$, that is the first $k$, such that $\left| \tilde{\mu}^{(k)} _t - \tilde{\mu}^{(k-1)} _t \right| < \delta$,
is used for $\tilde{\mu}^{(k)} _t$ to approximate $\mu_{t}$. In general, $\mu_{t}$ and its approximations are functions of portfolio vectors of two recent periods 
and the price relative vector,
\begin{equation} \label{eq:mu_as_func}
	\mu_{t} = \mu_{t}(\vec{w}_{t-1}, \vec{w}_{t}, \vec{y}_{t}).
\end{equation}
Throughout this work, a single constant commission rate for both selling and purchasing for
all non-cash assets is used, $c_\mathrm{s} = c_\mathrm{p} = 0.25\%$, the maximum rate at Poloniex.

The purpose of the algorithmic agent is to generate a time-sequence of portfolio vectors $\{ \vec  w_1, \vec  w_2, \cdots, \vec  w_t, \cdots \}$
in order to maximize the accumulative capital in \eqref{eq:final_P}, taking transaction cost into account. 

\subsection{Two Hypotheses}\label{section:hypotheses}
In this work, back-test tradings are only considered, where the trading agent pretends to be back in time
at a point in the market history, not knowing any "future" market information, and does paper trading from then onward.
As a requirement for the back-test experiments, the following two assumptions are imposed:
\begin{enumerate}
\item Zero slippage: The liquidity of all market assets is high enough that, each trade can be carried out immediately at the last price
	when a order is placed.
\item Zero market impact: The capital invested by the software trading agent is so insignificant that is has no influence on the market.
\end{enumerate}
In a real-world trading environment, if the trading volume in a market is high enough, these two assumptions are near to reality.

\section{Data Treatments} \label{section:data}
The trading experiments are done in the exchange Poloniex, where there are about 80 tradable cryptocurrency
pairs with about 65 available cryptocurrencies\footnote{as of May 23, 2017.}. 
However, for the reasons given below, only a subset of coins is considered by the trading robot in one period.
Apart from coin selection scheme, this section also gives a description of the data 
structure that the neural networks take as their input, a normalization pre-process, and a scheme to deal with missing data.

\subsection{Asset Pre-Selection} \label{section:preselection}
In the experiments of the paper, the 11 most-volumed non-cash assets are preselected for the portfolio. 
Together with the cash, Bitcoin, the size of the portfolio, $m+1$, is $12$. This number is chosen by experience and 
can be adjusted in future experiments. For markets with large volumes, like the foreign exchange market, $m$ can be as big as
the total number of available assets.

One reason for selecting top-volumed cryptocurrencies (simply called coins below) is that
bigger volume implies better market liquidity of an asset.
In turn it means the market condition is closer to Hypothesis 1 set in Section~\ref{section:hypotheses}.
Higher volumes also suggest that the investment can have less influence on the market, establishing an environment closer to the Hypothesis 2.
Considering the relatively high trading frequency (30 minutes) compared to some daily trading algorithms, liquidity and market size
are particularly important in the current setting.
In addition, the market of cryptocurrency is not stable.
Some previously rarely- or popularly-traded coins can have sudden boost or drop in volume in a short period of time.
Therefore, the volume for asset preselection is of a longer time-frame, relative to the trading period.
In these experiments, volumes of 30 days are used.

However, using top volumes for coin selection in back-test experiments can give rise to a \textit{survival bias}.
The trading volume of an asset is correlated to its popularity, which in turn is governed by its historic
performance. Giving future volume rankings to a back-test, will inevitably and indirectly pass future price information to
the experiment, causing unreliable positive results.
For this reason, volume information just before the beginning of the back-tests is taken for preselection
to avoid survival bias. 

\subsection{Price Tensor} \label{section:price_tensor}
Historic price data is fed into a neural network to generate the output of a portfolio vector. This subsection describes the structure
of the input tensor, its normalization scheme, and how missing data is dealt with.

The input to the neural networks at the end of Period $t$ is a tensor, $\bm{X}_t$, of rank 3 with shape $(f,n,m)$, where $m$ is the number of preselected non-cash assets,
$n$ is the number of input periods before $t$, and $f=3$ is the feature number. Since prices further back in the history have much less correlation
to the current moment than that of recent ones, $n=50$ (a day and an hour) for the experiments. The criterion of choosing the $m$ assets were given in
Section~\ref{section:preselection}. Features for asset $i$ on Period $t$ are its closing, highest,
and lowest prices in the interval. Using the notations from Section~\ref{section:math}, these are $v_{i,t}$, $v_{i,t}^\mathrm{(hi)}$,
and $v_{i,t}^\mathrm{(lo)}$.
However, these absolute price values are not directly fed to the networks. Since only the changes in prices will determine the performance
of the portfolio management (Equation~\eqref{eq:r_t}), all prices in the input tensor will be normalization by the latest
closing prices. Therefore, $\bm{X}_t$ is the stacking of the three normalized price matrices,
\begin{equation}
    \bm{X}_t = 
        \begin{tikzpicture}[baseline=-9ex,scale=1,every node/.style={draw,fill=white,opacity=0.8,minimum size=5em,inner sep=1}]
            \def\xs{1.85} 
            \def\ys{-0.0} 
    	    \node (mA)
    	        { $\bm{V}_t^\mathrm{(lo)}$ };
            \node (mB) at ($(mA.south west)+(\xs,\ys)$)
    	        { $\bm{V}_t^\mathrm{(hi)}$ };
            \node (mC) at ($(mB.south west)+(\xs,\ys)$)
    	        { $\bm{V}_t$ };
            \draw[dashed](mA.north east)--(mC.north east);
            \draw[dashed](mA.north west)--(mC.north west);
            \draw[dashed](mA.south west)--(mC.south west);
        \end{tikzpicture}.
    \label{eq:price_tensor}
\end{equation}
where $\bm{V}_t$, $\bm{V}_t^\mathrm{(hi)}$, and $\bm{V}_t^\mathrm{(lo)}$ are the normalized price matrices,
\begin{align}
    \bm{V}_t &= \bigg[
        \vec{v}_{t-n+1} \oslash  \vec{v}_{t} \Big|
        \vec{v}_{t-n+2} \oslash  \vec{v}_{t} \Big|
	\cdots \Big|
        \vec{v}_{t-1} \oslash  \vec{v}_{t} \Big|
	\vec{1}
    \bigg], \nonumber \\
    \bm{V}_t^\mathrm{(hi)} &= \bigg[
        \vec{v}_{t-n+1}^\mathrm{(hi)} \oslash  \vec{v}_{t} \Big|
        \vec{v}_{t-n+2}^\mathrm{(hi)} \oslash  \vec{v}_{t} \Big|
	\cdots \Big|
        \vec{v}_{t-1}^\mathrm{(hi)} \oslash  \vec{v}_{t} \Big|
        \vec{v}_{t}^\mathrm{(hi)} \oslash  \vec{v}_{t}
    \bigg], \nonumber \\
    \bm{V}_t^\mathrm{(lo)} &= \bigg[
        \vec{v}_{t-n+1}^\mathrm{(lo)} \oslash  \vec{v}_{t} \Big|
        \vec{v}_{t-n+2}^\mathrm{(lo)} \oslash  \vec{v}_{t} \Big|
	\cdots \Big|
        \vec{v}_{t-1}^\mathrm{(lo)} \oslash  \vec{v}_{t} \Big|
        \vec{v}_{t}^\mathrm{(lo)} \oslash  \vec{v}_{t}
    \bigg], \nonumber
\end{align}
with $\vec{1} = \left( 1,1,\cdots,1  \right)^\intercal$, and $\oslash$ being the element-wise division operator.

At the end of Period $t$, the portfolio manager comes up with a portfolio vector $\vec{w}_t$ using merely the information from the price tensor $\bm{X}_t$
and the previous portfolio vector $\vec{w}_{t-1}$,
according to some policy $\pi$. In other words, $\vec{w}_t = \pi(\bm{X}_t, \vec{w}_{t-1})$. At the end of Period $t+1$, the logarithmic rate of return
for the period due to 
decision $\vec{w}_t$ can be calculated with the additional information from the price change vector $\vec{y}_{t+1}$, using Equation~\eqref{eq:r_t}, 
$r_{t+1} = \ln( \mu_{t+1} \vec{y}_{t+1} \bcdot \vec{w}_t ).$
In the language of RL, $r_{t+1}$ is the immediate reward to the portfolio management agent for its action $\vec{w}_t$ under environment condition $\bm{X}_t$.

\subsection{Filling Missing Data} \label{section:missing_data}
Some of the selected coins lack part of the history. This absence of data is due to the fact that these coins just appeared relatively recently.
Data points before the existence of a coin are marked as Not A Numbers (NANs) from the exchange. NANs only appeared in the training set,
because the coin selection criterion is the volume-ranking of the last 30 days before the back-tests, meaning all assets must have existed before that.

As the input of a neural network must be real numbers, these NANs have to be replaced. 
In a previous work of the authors \citep{jiang2016}, the missing data was filled with fake decreasing price series with a decay rate of 0.01,
in order for the neural networks to avoid picking these absent assets in the training process.
However, it turned out that the networks deeply remembered these particular assets, that they avoided them even when they were in very promising up-climbing 
trends in the back-test experiments.
For this reason, in this current work, flat fake price-movements (0 decay rates) are used to fill the missing data points.
In addition, under the novel EIIE structure, the new networks will not be able to reveal the identity of individual assets, preventing
them from making decision based on the long-past bad records of particular assets.

\section{Reinforcement Learning}\label{section:rl}
With the problem defined in Section~\ref{section:problem} in mind, this section presents a reinforcement-learning (RL) solution framework using a 
deterministic policy gradient algorithm. The explicit reward function is also given under this framework.

\subsection{The Environment and the Agent}
In the problem of algorithmic portfolio management, the agent is the software portfolio manager performing trading-actions in the environment of a financial market. 
This environment comprises of all available assets in the markets and the expectations of all market participants towards them.

It is impossible for the agent to get total information of a state of such a large and complex environment. Nonetheless, all relevant information is believed,
in the philosophy of technical traders \cite{charles2006technical,lo2000foundations}, to be reflected in the prices of the assets,
which are publicly available to the agent. Under this point of view, an environmental state can be roughly represented by the prices of all orders throughout
the market's history up to the moment where the state is at. Although full order history is in the public domain for many financial markets,
it is too huge a task for the software agent to practically process this information. As a consequence, sub-sampling schemes for the order-history
information are employed to future simplify the state representation of the market environment. These schemes include asset preselection described in 
Section~\ref{section:preselection}, periodic feature extraction and
history cut-off. Periodic feature extraction discretizes the time into periods, and then extract the highest, lowest, and closing prices in each periods. 
History cut-off simply takes the price-features of only a recent number of periods to represent the current state of the environment. The resultant 
representation is the price tensor $\bm{X}_t$ described in Section~\ref{section:price_tensor}.

Under Hypothesis 2 in Section~\ref{section:hypotheses}, the trading action of the agent will not influence the future price states of the market.
However, the action made at the beginning of Period $t$ will affect the reward of Period $t+1$, and as a result will affect the decision of its action.
The agent's buying and selling transactions made at the beginning of Period $t+1$, aiming to redistribute the wealth among the assets, 
are determined by the difference between portfolio weights $\vec w'_{t}$ and $\vec{w}_t$.  $\vec w'_{t}$ is defined in term of $\vec w_{t-1}$
in Equation~\eqref{eq:w_prime}, which also plays a role in the action for the last period. 
Since $\vec w_{t-1}$ has already been determined in the last period
the action of the agent at time $t$ can be represented solely by the portfolio vector $\vec w_t$,
\begin{equation}  \label{eq:action}
	\vec a_t = \vec{w_t}.
\end{equation}
 Therefore a previous action does
 have influence on the decision of the current one through the dependency of $r_{t+1}$ and $\mu_{t+1}$ on $\vec w_{t}$ \eqref{eq:mu_as_func}.
In the current framework, this influence is encapsulated by considering $\vec w_{t-1}$ as a part of the environment and inputting it
to the agent's action making policy, so the state at $t$ is represented as the pair of $X_t$ and $\vec w_{t-1}$,
\begin{equation} \label{eq:state}
	\vec s_t = (\bm X_t, \vec w_{t-1}),
\end{equation}
where $\vec w_0$ is predetermined in \eqref{eq:w_0}. The state $\bm s_t$  consists of two parts, the external state represented by the price tensor,
$\bm X_t$, and the internal state represented by the portfolio vector from the last period, $\vec w_{t-1}$. Because under Hypothesis 2 of
Section~\ref{section:hypotheses}, the portfolio amount is negligible compared to the total trading volume of the market, $p_t$ is not included in
the internal state. 

\subsection{Full-Exploitation and the Reward Function}
It is the job of the agent to maximize the final portfolio value $p_\mathrm{f}$ of Equation~\eqref{eq:final_P} at the end of the $t_\mathrm{f}+1$ period.
As the agent does not have control over the choices of the initial investment, $p_0$, and the length of the whole portfolio management process, $t_\mathrm{f}$,
this job is equivalent to maximizing the average logarithmic cumulated return $R$,
\begin{align}
	R(\vec s_1,\vec a_1,\cdots,\vec s_{t_\mathrm{f}},\vec a_{t_\mathrm{f}},\vec s_{t_\mathrm{f}+1})
	&:= \frac{1}{t_\mathrm{f}} \ln \frac{p_\mathrm{f}}{p_0}
	= \frac{1}{t_\mathrm{f}} \sum_{t=1}^{t_\mathrm{f}+1} \ln \left(\mu_t \vec y_{t} \bcdot \vec w_{t-1} \right) \label{eq:mean_log_ret} \\
	&= \frac{1}{t_\mathrm{f}} \sum_{t=1}^{t_\mathrm{f}+1} r_{t}. \label{eq:imm_ret}
\end{align}
On the right-hand side of \eqref{eq:mean_log_ret}, $\vec w_{t-1}$ is given by action $\vec a_{t-1}$,  $\vec y_{t}$ is part of price tensor
$\bm X_{t}$ from state variable $\vec s_{t}$, and $\mu_t$ is a function of $\vec w_{t-1}$, $\vec w_{t}$ and $\vec y_{t}$ as stated in \eqref{eq:mu_as_func}.
In the language of RL, $R$ is the cumulated reward, and $r_t/t_\mathrm{f}$ is the immediate reward for an individual episode.
Different from a reward function using accumulated portfolio value \cite{moody1998performance}, the denominator 
$t_\mathrm{f}$ guarantees the fairness of the reward function between runs of different lengths, enabling it to train the trading policy in mini-batches.

With this reward function, the current framework has two important distinctions from many other RL problems. One is that both the episodic and 
cumulated rewards are exactly expressed. In other words, the domain knowledge of the environment is well-mastered, and can be fully exploited by the agent.
This exact expressiveness is based upon
Hypothesis 1 of Section~\ref{section:hypotheses} that an action has no influence on the external part of future states, the price tensor. This isolation
of action and external environment also allows one to use the same segment of market history to evaluate difference sequences of actions. This feature of 
the framework is considered a major advantage, because a complete new trial in a trading game is both time-consuming and expansive.

The second distinction is that all episodic rewards are equally important to the final return. This distinction, together with the zero-market-impact assumption, 
allows $r_t/t_\mathrm{f}$ to be regarded as the action-value function of action $\vec w_t$ with a discounted factor of $0$, taking no consideration
of future influence of the action.  Having a definite action-value function further justifies the full-exploitation approach, since exploration in other 
RL problems is mainly for trying out different classes of action-value functions.

Without exploration, on the other hand, local optima can be avoided by random initialisation of the policy parameters which will be discussed below.

\subsection{Deterministic Policy Gradient}
A policy is a mapping from the state space to the action space, $\pi:\mathcal{S}\rightarrow\mathcal{A}$. With full exploitation in the current framework,
an action is deterministically produced by the policy from a state. The optimal policy is obtained using a gradient ascent algorithm. To achieve this,
a policy is specified by a set of parameter $\vec \theta$, and $\vec a_t = \pi_{\vec \theta}(\bm s_t)$. The performance
metric of $\pi_{\vec \theta}$ for time interval $[0,t_\mathrm{f}]$ is defined as the corresponding reward function \eqref{eq:mean_log_ret} of the interval,
\begin{equation}
	J_{[0,t_\mathrm{f}]}(\pi_{\vec \theta}) = R\left( \vec s_1,\pi_{\vec \theta}(s_1),\cdots,
		\vec s_{t_\mathrm{f}},\pi_{\vec \theta}(s_{t_\mathrm{f}}),\vec s_{t_\mathrm{f}+1} \right).
	\label{eq:policy_value}
\end{equation}
After random initialisation, the parameters are continuously updated along the gradient direction with a learning rate $\lambda$,
\begin{equation}
	\vec\theta \longrightarrow \vec\theta + \lambda\nabla_{\vec\theta}J_{[0,t_\mathrm{f}]}(\pi_{\vec \theta}).
	\label{eq:gradient_ascent}
\end{equation}
To improve training efficiency and avoid machine-precision errors, $\vec\theta$ will be updated upon mini-batches instead of the whole training market-history.
If the time-range of a mini-batch is $[t_{\mathrm{b}_1}, t_{\mathrm{b}_2}]$, the updating rule for the batch is
\begin{equation}
	\vec\theta \longrightarrow \vec\theta + \lambda\nabla_{\vec\theta}J_{[t_{\mathrm{b}_1}, t_{\mathrm{b}_2}]}(\pi_{\vec \theta}),
\end{equation}
with the denominator in the corresponding $R$ defined in \eqref{eq:mean_log_ret} replaced by $t_{\mathrm{b}_2} - t_{\mathrm{b}_1}$.
This mini-batch approach of gradient ascent also allows online learning, which is important in online trading where new market history keep coming to the agent.
Details of the online learning and mini-batch training will be discussed in Section~\ref{section:training}

\section{Policy Networks} \label{section:networks}

The policy functions $\pi_{\vec\theta}$ will be constructed using three different deep neural networks.
The neural networks in this paper differ from a previous version \cite{jiang2016} with three important innovations,
the mini-machine topology invented to target the portfolio management problem, the portfolio-vector memory,
and a stochastic mini-batch online learning scheme.

\subsection{Network Topologies} \label{section:topology}

The three incarnations of neural networks to build up the policy functions are a CNN, a basic RNN, and a LSTM. Figure~\ref{fig:cnn_topology}
shows the topology of a CNN designed for solving the current portfolio management problem, while Figure~\ref{fig:rnn_topology} portrays the structure
of a basic RNN or LSTM network for the same problem. In all cases, the input to the networks is the price tensor $\bm{X}_t$ defined in
\eqref{eq:price_tensor}, and the output is the portfolio vector $\vec w_t$. In both figures, an hypothetical example of output portfolio vector is used,
while the dimension of the price tensor and thus the number of assets are actual values deployed in the experiments. 
The last hidden layers are the voting scores for all non-cash assets. The softmax outcomes of these scores and a cash bias
become the actual corresponding portfolio weights. In order for the neural network to consider transaction cost, the portfolio vector from the last
period, $\vec w_{t-1}$, is inserted to the networks just before the voting-layer. The actual mechanism of storing and retrieving portfolio vectors
in a parallel manner is presented in Section \ref{section:memory}.

A vital common feature in all three networks is that the networks flow independently for the $m$ assets while network parameters are shared among these
streams. These streams are like independent but identical networks of smaller scopes, separately observing and assessing individual non-cash assets. 
They only interconnect at the softmax function, just to make sure their outputting weights are non-negative and summing up to unity. We call these
streams mini-machines or more formally Identical Independent Evaluators (IIE), and this topology feature Ensemble of IIE (EIIE) nicknamed mini-machine approach,
to distinguish with the wholesome approach in an earlier attempt \cite{jiang2016}.
EIIE is realized differently in Figure~\ref{fig:cnn_topology} and \ref{fig:rnn_topology}. An IIE in Figure~\ref{fig:cnn_topology} is just a chain
of convolution with kernels of height $1$, while in Figure~\ref{fig:rnn_topology} it is either a LSTM or a Basic RNN taking the price history
of a single asset as input.

EIIE greatly improves the performance of the portfolio management. Remembering the historic performance
of individual assets, an integrated network in the previous version is more reluctant to invest money to a historically unfavorable asset, even if
the asset has a much more promising future. On the other hand, without being designed to reveal  
the identity of the assigned asset, an IIE is able to judge its potential rise and fall merely based on more recent events.

From a practical point of view, EIIE
has three other crucial advantages over an integrated network. The first is scalability in asset number.
Having the mini-machines all identical with shared parameters, the training time of an ensemble scales roughly linearly with $m$.
The second advantage is data-usage efficiency. For an interval of price history, a 
mini-machine can be trained $m$ times across different assets. Asset assessing experience of the IIEs is then shared and accumulated in both time and asset dimensions.
The final advantage is plasticity to asset collection. Since an IIE's asset assessing ability is universal without being restricted to any particular
assets, an EIIE can update its choice of assets and/or the size of the portfolio in real-time, without having to train the network again from ground zero.

\begin{figure}[!htb]
    \centering
    \includegraphics[width=\linewidth]{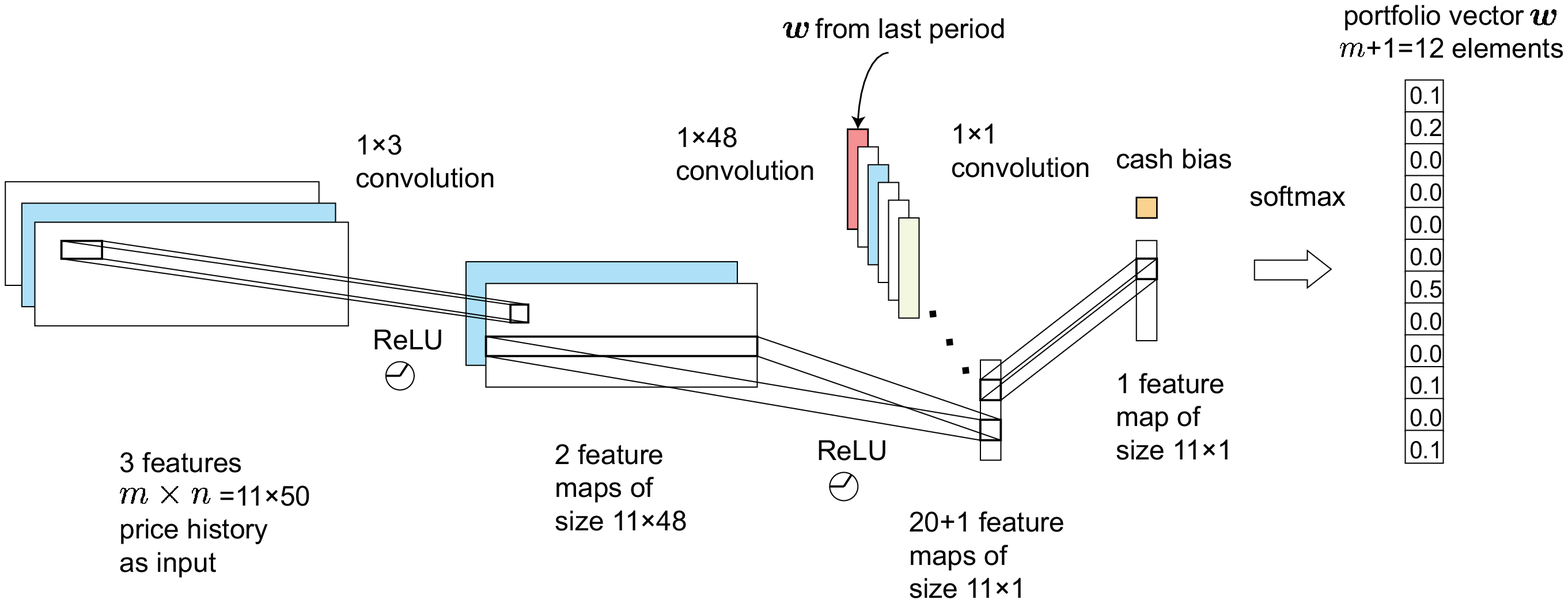}
    \caption{CNN Implementation of the EIIE: This is a realization the Ensemble of Identical Independent Evaluators (EIIE),
    a fully convolutional network. The first dimensions of all the local receptive fields
    in all feature maps are $1$, making all rows isolated from each other until the softmax activation. Apart from 
    weight-sharing among receptive fields in a feature map, which is a usual CNN characteristic, parameters are also shared between rows in an EIIE
    configuration. Each row of the entire network is assigned with a particular asset, and is responsible to submit a voting
    score to the softmax on the growing potential of the asset in the coming trading period. The input to the network is a
    $3\times m\times n$ price tensor, comprising the highest, closing, and lowest prices of $m$ non-cash assets over the past
    $n$ periods. The outputs are the new portfolio weights. The previous portfolio weights are inserted as an extra 
    feature map before the scoring layer, for the agent to minimize transaction cost.
    }
    \label{fig:cnn_topology}
\end{figure}

\begin{figure}[!htb]
    \centering
    \includegraphics[width=\linewidth]{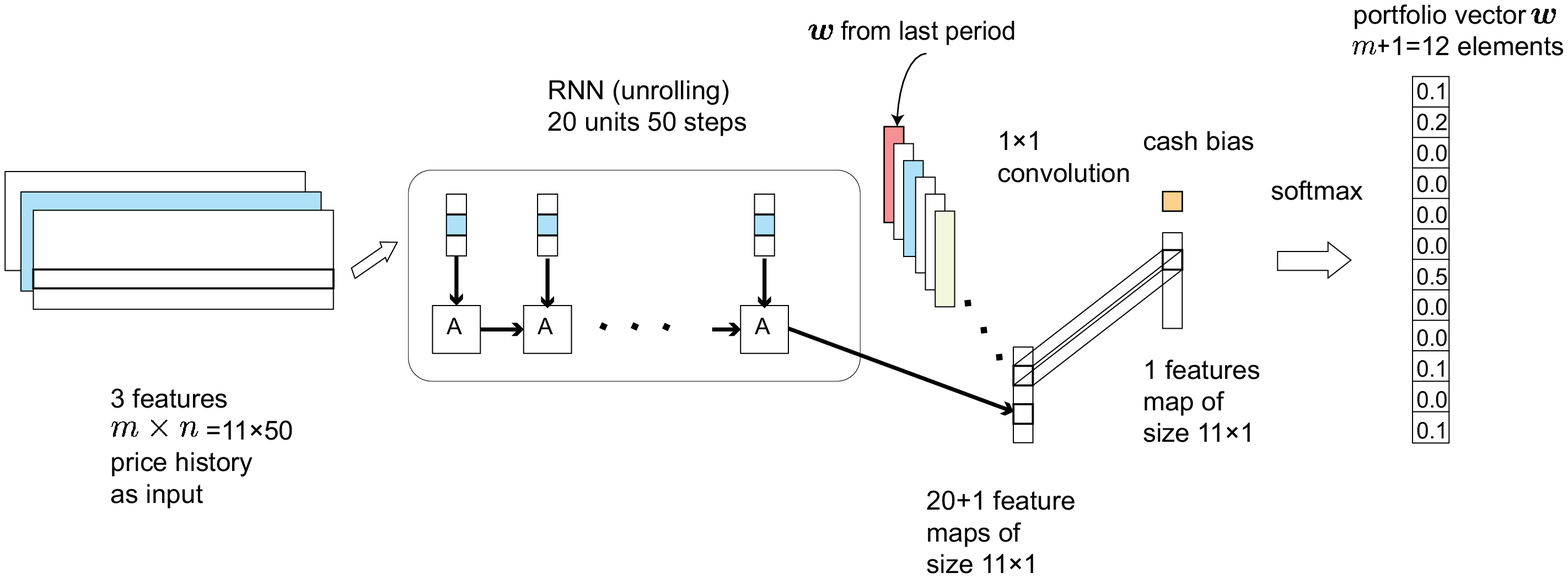}
    \caption{RNN (Basic RNN or LSTM) Implementation of the EIIE:
    This is a recurrent realization the Ensemble of Identical Independent Evaluators (EIIE).
    In this version, the price inputs of individual assets are taken by small recurrent subnets. These subnets
    are identical LSTMs or Basic RNNs. The structure of the ensemble network after the recurrent subnets is the same as
    the second half of the CNN in Figure \ref{fig:cnn_topology}.  
    }
    \label{fig:rnn_topology}
\end{figure}

\subsection{Portfolio-Vector Memory} \label{section:memory}
\begin{figure}[!htb]
	\centering
	\subfigure[Mini-Batch Viewpoint]{
		\includegraphics[width=.47\linewidth]{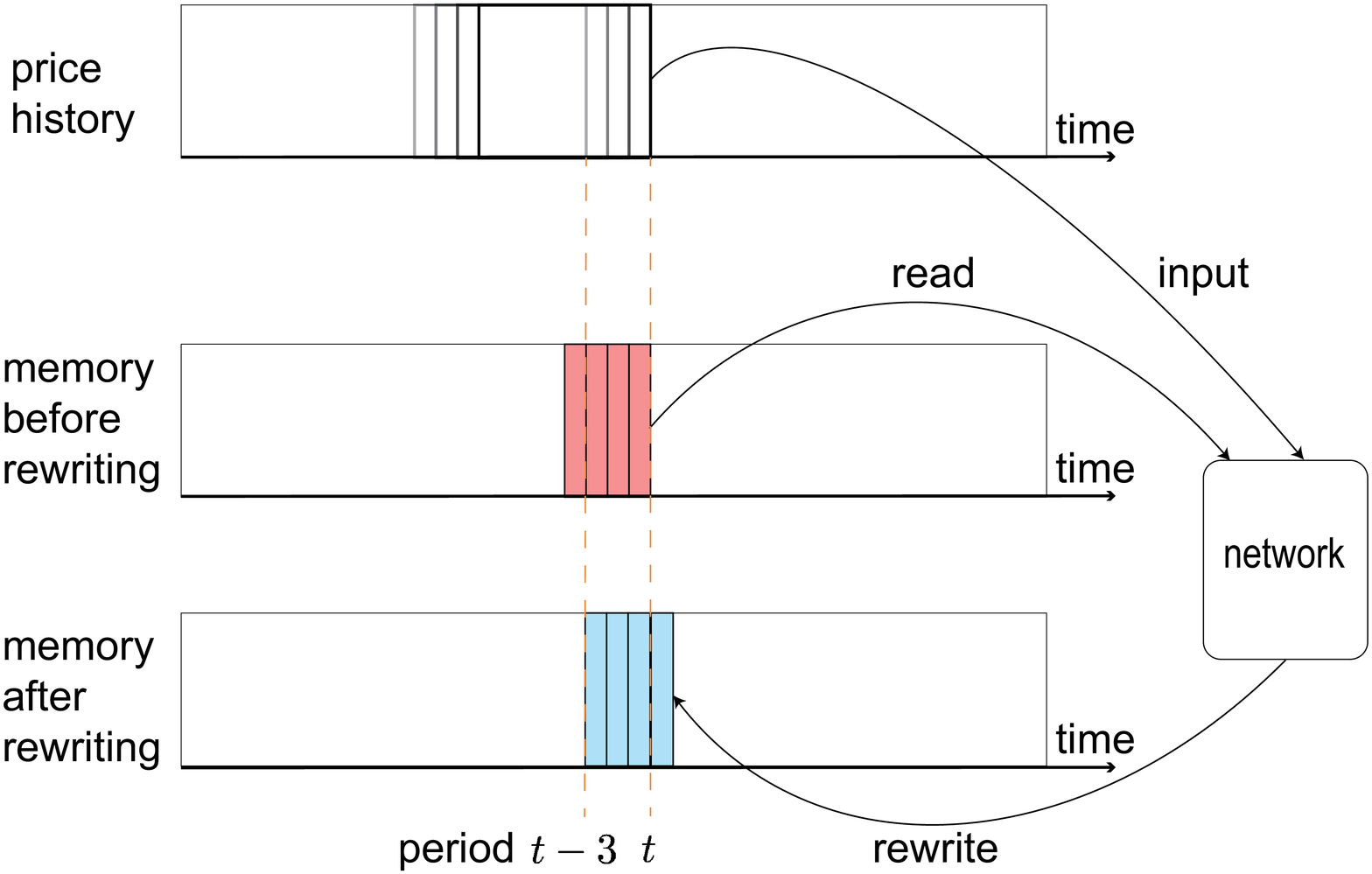}
		\label{fig:macro}
	}
	\subfigure[Network Viewpoint]{
		\includegraphics[width=.47\linewidth]{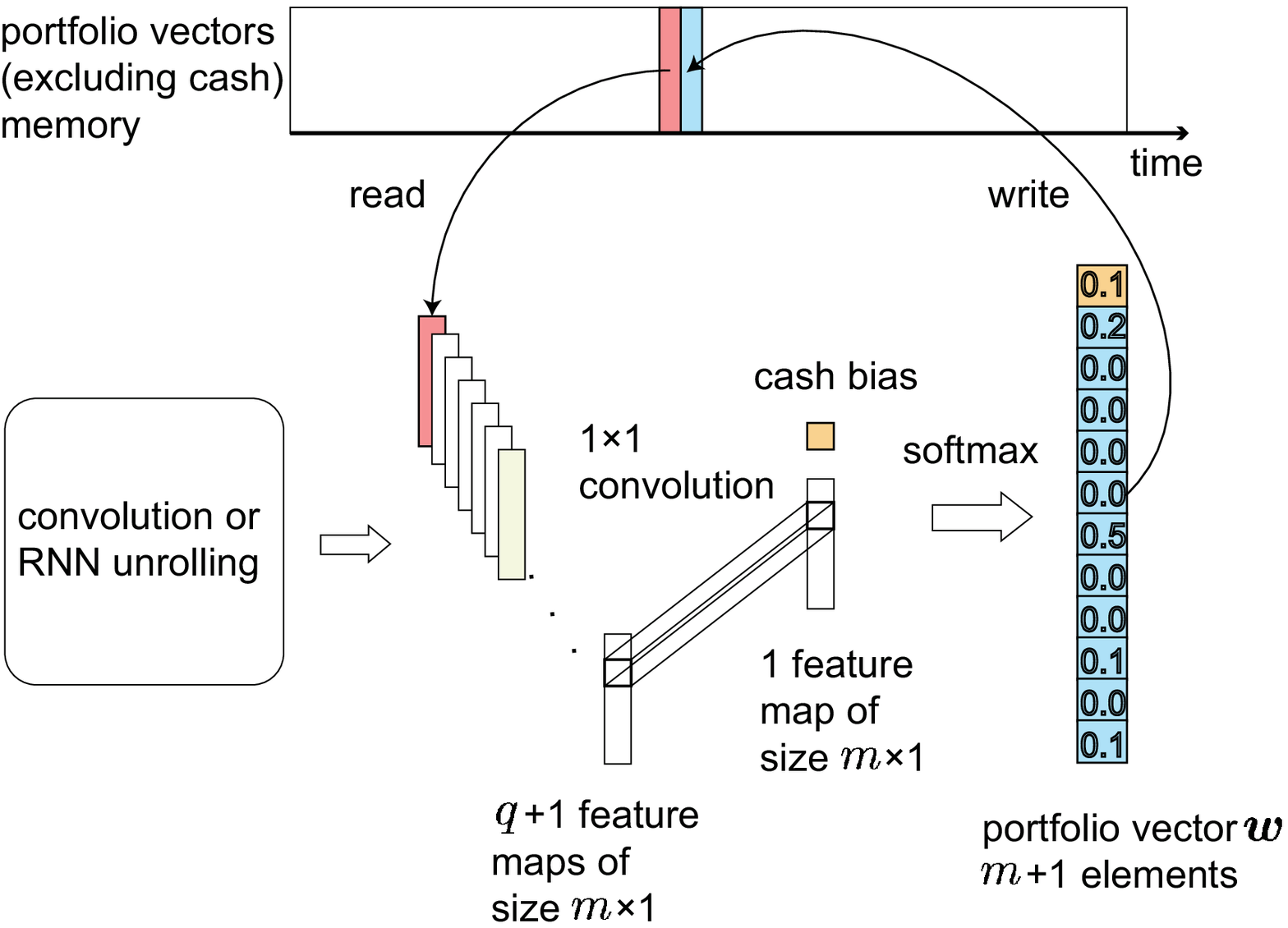}
		\label{fig:micro}
	}
	\caption{A Read/Write Cycle of the Portfolio-Vector Memory: In both graphs, a small vertical strip on the time axis represents a portion of
		the memory containing the portfolio weights at the beginning of a period. Red memories are being read to the policy network, while
		blue ones are being overwritten by the network. The two colored rectangles in \subref{fig:macro} consisting of four strips are
		example of two consecutive mini-batches. While \subref{fig:macro} exhibits a complete read and write circle for a mini-batch,
		\subref{fig:micro} shows a circle within a network (omitting the CNN or RNN part of the network).}
	\label{fig:last_w}
\end{figure}
In order for the portfolio management agent to minimize transaction cost by restraining itself from large changes between consecutive portfolio vectors, 
the output of portfolio weights from the previous trading period is input to the networks. One way
to achieve this is to rely on the remembering ability of RNN, but with this approach the price normalization scheme proposed in \eqref{eq:price_tensor} has
to be abandoned. This normalization scheme is empirically better performing than others. Another possible solution is Direct Reinforcement (RR)
introduced by \citet{moody2001drl}. However, both RR and RNN memory suffer from the gradient vanishing problem. More importantly, RR and RNN require
serialization of the training process, unable to utilize parallel training within mini-batches.

In this work, inspired by the idea of experience replay memory \cite{mnih2016asynchronous}, a dedicated Portfolio-Vector Memory (PVM),
is introduced to store the network outputs. As shown in Figure~\ref{fig:last_w}, the PVM
is a stack of portfolio vectors in chronological order. Before any network training, the PVM is initialized with uniform weights. In each training
step, a policy network loads the portfolio vector of the previous period from the memory location at $t-1$, and overwrites the memory at $t$ with its
output. As the parameters of the policy networks converge through many training epochs, the values in the memory also converge.

Sharing a single memory stack allows a network to be trained simultaneously against data points within a mini-batches, enormously improving training efficiency.
In the case of RNN versions of the networks, inserting last outputs after the recurrent blocks (Figure~\ref{fig:rnn_topology}) avoids passing the
gradients back to the deep RNN structures, circumventing the gradient vanishing problem.

\subsection{Online Stochastic Batch Learning} \label{section:training}
With the introduction of the network output-memory, mini-batch training becomes plausible, although the learning framework requires sequential inputs.
However, unlike supervised learning, where data points are unordered and mini-batches are random disjoint subsets of the training sample space,
in this training scheme the data points within a batch have to be in their time-order. In addition, since data sets are time series,
mini-batches starting with different periods are considered valid and distinctive, even if they have a significantly overlapping interval. For example, if
the uniform batch size is $n_\mathrm{b}$, data sets covering $[t_\mathrm{b}, t_\mathrm{b} + n_\mathrm{b})$ and  
$[t_\mathrm{b}+1, t_\mathrm{b} + n_\mathrm{b}+1)$ are two validly different batches.

The ever-ongoing nature of financial markets means new data keeps pouring into the agent, and as a consequence the size of the of training sample explodes
indefinitely. Fortunately, it is believed that the correlation between two market price events decades exponentially with the temporal distance between them
\cite{holt2004forecasting,charles2006technical}. With this belief, here an Online Stochastic Batch Learning (OSBL) scheme is proposed.

At the end of the $t$th period, the price movement of this period will be added to the training set. After the agent has completed its orders for period $t+1$,  
the policy network will be trained against $N_\mathrm{b}$ randomly chosen mini-batches from this set. A batch starting with period
$t_\mathrm{b} \leqslant t-n_\mathrm{b}$ is picked with a geometrically distributed probability $P_\beta(t_\mathrm{b})$,
\begin{equation} \label{eq:prob_batch}
	P_\beta(t_\mathrm{b}) = \beta (1-\beta)^{t - t_\mathrm{b} - n_\mathrm{b}},
\end{equation}
where $\beta\in(0,1)$ is the probability-decaying rate determining the shape of the probability distribution and how important are recent market events, and
$n_\mathrm{b}$ is the number of periods in a mini-batch.

\section{Experiments} \label{section:experiments}

The tools has been developed to this point of the article are examined in three back-test experiments of different time frames with all three policy networks on the
crypto-currency exchange Poloniex. Results are compared with many well-established and recently published portfolio-selection strategies. The main compared
financial metric is the portfolio value as well as maximum drawdown and the Sharpe ratio.

\subsection{Test Ranges} \label{section:ranges}
\begin{table}[!htb]
    \centering
    \begin{tabular}{c|c|c} 
    Data Purpose & Data Range & Training Data Set \\
    \hline
    CV & 2016-05-07 04:00 to 2016-06-27 08:00& 2014-07-01 to 2016-05-07 04:00 \\
    \hline
    Back-Test 1 & 2016-09-07 04:00 to 2016-10-28 08:00 & 2014-11-01 to 2016-09-07 04:00 \\
    \hline
    Back-Test 2 & 2016-12-08 04:00 to 2017-01-28 08:00 & 2015-02-01 to 2016-12-08 04:00 \\
    \hline
    Back-Test 3 & 2017-03-07 04:00 to 2017-04-27 08:00 & 2015-05-01 to 2017-03-07 04:00 \\
    \end{tabular} 
    \caption{Price data ranges for hyperparameter-selection (cross-validation, CV) and back-test experiments. Prices are accessed in periods of 30 minutes. 
    Closing prices are used for cross validation and back-tests, while highest, lowest, and closing prices in the periods are used for training. 
    The hours of the starting points for the training sets are not given, since they begin at midnight of the days.  All times are in
    UTC.}
    \label{tab:datarange}
\end{table} 
Details of the time-ranges for the back-test experiments and their 
corresponding training sets are presented in Table~\ref{tab:datarange}.  A cross validation set is used for determination of the hyper-parameters,
whose range is also listed. All time in the table are in Coordinated Universal Time (UTC). All training sets start at 0 o'clock. For example,
the training set for Back-Test 1 is from 00:00 on November 1st 2014. 
All price data is accessed with Poloniex's official Application Programming Interface (API)\footnote{https://poloniex.com/support/api/}.

\subsection{Performance Measures}
Different metrics are used to measure the performance of a particular portfolio selection strategy. The most direct measurement of how successful is a portfolio
management over a timespan is the accumulative portfolio value (APV), $p_t$. It is unfair, however, to compare the PVs of two management starting of different
initial values. Therefore, APVs here are measured in the unit of their initial values, or equivalently $p_0=1$ and thus
\begin{equation} \label{eq:fapv}
	p_t = p_t/p_0.
\end{equation}
In this unit, APV is then closely related to the accumulated return, and in fact it only differs from the latter by $1$. 
Under the same unit, the final APV (fAPV) is the APV at the end of a back-test experiment, $p_\mathrm{f} = p_\mathrm{f}/p_0 = p_{t_\mathrm{f}+1} / p_0$.

A major disadvantage of APV is that it does not measure the risk factors, since it merely sums up all the periodic returns without considering fluctuation in
these returns. A second metric, the Sharpe ratio (SR) \cite{sharpe1964,sharpe1994}, is used to take risk into account. The ratio is a risk adjusted mean return,
defined as the average of the risk-free return by its deviation,
\begin{equation} \label{eq:sharpe}
	S = \frac{\mathbb{E}_{t}[ \rho_t - \rho_\mathrm{F}]}
	{\sqrt{\var\limits_{t}\left( \rho_t - \rho_\mathrm{F} \right)}},
\end{equation}
where $\rho_t$ are periodic returns defined in \eqref{eq:rho_t}, and $\rho_\mathrm{F}$ is the rate of return of a risk-free asset. In these experiments
the risk-free asset is Bitcoin. Because the quoted currency is also Bitcoin, the risk-free return is zero, $\rho_\mathrm{F} = 0$, here.

Although the SR considers volatility of the portfolio values, but it equally treats upwards and downwards movements. In reality upwards volatility
contributes to positive returns, but downwards to loss. In order to highlight the downwards deviation, Maximum Drawdown (MDD) \cite{magdon2004maximum} is
also considered. MDD is the biggest loss from a peak to a trough, and mathematically
\begin{equation} \label{eq:mdd}
	D = \max_{\substack{\tau>t\\ t}} \frac{p_t-p_\tau}{p_t}.
\end{equation}

\subsection{Results}

\begin{table}[!htb]
	\centering
		\small
		\begin{tabular}{lV{3}cccV{3}cccV{3}ccc}
			{} &\multicolumn{3}{cV{3}}{2016-09-07 to 2016-10-28}
			& \multicolumn{3}{cV{3}}{2016-12-08 to 2017-01-28}
			& \multicolumn{3}{c}{2017-03-07 to 2017-04-27} \\
		\hlineB{3}     
			{}Algorithm &MDD & fAPV & SR &MDD & fAPV & SR &MDD & fAPV & SR\\
		\hlineB{3}     
		\bf{CNN}                                               
		&0.224&\bf{29.695}&\bf{0.087}&0.216&\bf{8.026}&\bf{0.059}&0.406&31.747&0.076 \\
		\bf{bRNN}
		&0.241&13.348&0.074&0.262&4.623&0.043&0.393&\bf{47.148}&\bf{0.082}\\
		\bf{LSTM}
		&0.280&6.692&0.053&0.319&4.073&0.038&0.487&21.173&0.060 \\
		iCNN                                               
		&\bf{0.221}&4.542&0.053&
		0.265&1.573&0.022
		&0.204&3.958&0.044 \\
		\hline
		\it{Best Stock}
		&0.654&1.223&0.012&0.236&1.401&0.018&0.668&4.594&0.033 \\
		\it{UCRP}
		&0.265&0.867&-0.014&\bf{0.185}&1.101&0.010&\bf{0.162}&2.412&0.049 \\
		\it{UBAH}
		&0.324&0.821&-0.015&0.224&1.029&0.004&0.274&2.230&0.036 \\
		\hline
		Anticor
		&0.265&0.867&-0.014&0.185&1.101&0.010&0.162&2.412&0.049 \\
		OLMAR
		&0.913&0.142&-0.039&0.897&0.123&-0.038&0.733&4.582&0.034\\
		PAMR
		&0.997&0.003&-0.137&0.998&0.003&-0.121&0.981&0.021&-0.055 \\
		WMAMR
		&0.682&0.742&-0.0008&0.519&0.895&0.005&0.673&6.692&0.042 \\
		CWMR
		&0.999&0.001&-0.148&0.999&0.002&-0.127&0.987&0.013&-0.061 \\
		RMR
		&0.900&0.127&-0.043&0.929&0.090&-0.045&0.698&7.008&0.041 \\
		\hline
		ONS
		&0.233&0.923&-0.006&0.295&1.188&0.012&0.170&1.609&0.027\\
		UP
		&0.269&0.864&-0.014&0.188&1.094&0.009&0.165&2.407&0.049 \\
		EG
		&0.268&0.865&-0.014&0.187&1.097&0.010&0.163&2.412&0.049 \\
		\hline
		$\mathrm{B^K}$
		&0.436&0.758&-0.013&0.336&0.770&-0.012&0.390&2.070&0.027 \\
		CORN
		&0.999&0.001&-0.129&1.000&0.0001&-0.179&0.999&0.001&-0.125 \\
		M0
		&0.335&0.933&-0.001&0.308&1.106&0.008&0.180&2.729&0.044\\
	\end{tabular}
	\caption[Performances]{Performances of the three EIIE (Ensemble of Identical Independent Evaluators)
	neural networks, an integrated network, and some traditional portfolio selection strategies in three different back-test experiments
	(in UTC, detailed time-ranges listed in Table~\ref{tab:datarange})
	on the cryptocurrency exchange Poloniex.
	The performance metrics are Maximum Drawdown (MDD),
	the final Accumulated Portfolio Value (fAPV) in the unit of initial portfolio amount ($p_\mathrm{f} / p_0$),
	and the Sharpe ratio (SR). The bold algorithms are the EIIE networks introduced in this paper, named
	after the underlining structures of their IIEs. For example, bRNN is the EIIE of Figure~\ref{fig:rnn_topology} using basic RNN evaluators.
	Three benchmarks (italic), the integrated CNN (iCNN) previous proposed by the authors \cite{jiang2016}, and some recently
	reviewed\textsuperscript{a} strategies \cite{li2015olps,li2014survey}
	are also tested.
	The algorithms in the table are divided into five categories, the model-free neural network, the benchmarks, follow-the-loser strategies,
	follow-the-winner strategies, and pattern-matching or other strategies.
	The best performance in each column is highlighted
	with boldface. All three EIIEs
	significantly outperform all other algorithms in the fAPV and SR columns, showing the profitability and reliability of the EIIE machine-learning solution
	to the portfolio management problem.
	}
	\footnotesize a. The exceptions are RMR of \citet{huang2013robust} and WMAMR of \citet{gao2013weighted}.
	\label{tab:performance}
\end{table}

The performances of all three EIIE policy networks proposed in the current paper will be compared to that of the integrated CNN (iCNN) \cite{jiang2016},
several well-known or recently published model-based strategies, and three benchmarks. 

The three benchmarks are the Best Stock, the asset with the most fAPV over the back-test interval, the Uniform Buy and Hold (UBAH), a portfolio management
approach simply equally spreading the total fund into the preselected assets and holding them without making any purchases or selling until the
end \cite{li2014survey}, and Uniform Constant Rebalanced Portfolios (UCRP) \cite{kelly1956crp,cover1991universal}.

Most of the strategies to be compared in this work were surveyed by \citet{li2014survey}, including Aniticor \cite{borodin2004can}, Online Moving Average Reversion 
(OLMAR) \cite{li2015moving}, Passive Aggressive Mean Reversion (PAMR) \cite{li2012pamr}, 
Confidence Weighted Mean Reversion (CWMR) \cite{li2013confidence},
Online Newton Step (ONS) \cite{agarwal2006algorithms}, Universal Portfolios (UP) \cite{cover1991universal}, 
Exponential Gradient (EG) \cite{helmbold1998line}, Nonparametric Kernel Based Log Optimal Strategy ($\mathrm{B^K}$) \cite{gyorfi2006nonparametric},
Correlation-driven Nonparametric Learning Strategy (CORN) \cite{li2011corn}, and M0 \cite{borodin2000competitive}, 
except Weighted Moving Average Mean Reversion (WMAMR) \cite{gao2013weighted} and Robust Median Reversion (RMR) \cite{huang2013robust}.

Table~\ref{tab:performance} shows the performance scores fAPV, SR, and MDD of the EIIE policy networks as well as of the compared strategies for
the three back-test intervals listed in Table~\ref{tab:datarange}. In term of fAPV or SR, the best performing algorithm in Back-Test 1 and 2 is the CNN EIIE
whose final wealth is more than twice of the runner-up in the first experiment. Top three winners in these two measures in all back-tests are occupied by
the three EIIE networks, losing only the MDD measure. This result demonstrates the powerful profitability and consistency of the current EIIE machine-learning
framework.

When only considering fAPV, all three EIIEs outperform the best assets in all three back-tests, while the only model-based algorithm does that is RMR on
the only occasion of Back-Test 3. Because of the high commission rate of $0.25\%$ and the relatively high half-hourly trading frequency, many traditional
strategies have bad performances. Especially in Back-Test 1, all model-based strategies have negative returns, with fAPV less than 1 or equivalently negative SRs.
On the other hand, the EIIEs are able to achieve at least 4-fold returns in 20 days in different market conditions. 

Figures \ref{fig:backtest1}, \ref{fig:backtest2} and \ref{fig:backtest3} plot the APV against time in the three back-tests respectively for the
CNN and bRNN EIIE networks, two selected benchmarks and two model-based strategies. The benchmarks Best Stock and UCRP are two good representatives of the market.
In all three experiments, both CNN and bRNN EIIEs beat the market throughout the entirety of the back-tests, while traditional strategies are only able to achieve
that in the second half of Back-Test 3 and very briefly elsewhere.

\begin{figure}[!h]
\centering
\includegraphics[width=\linewidth]{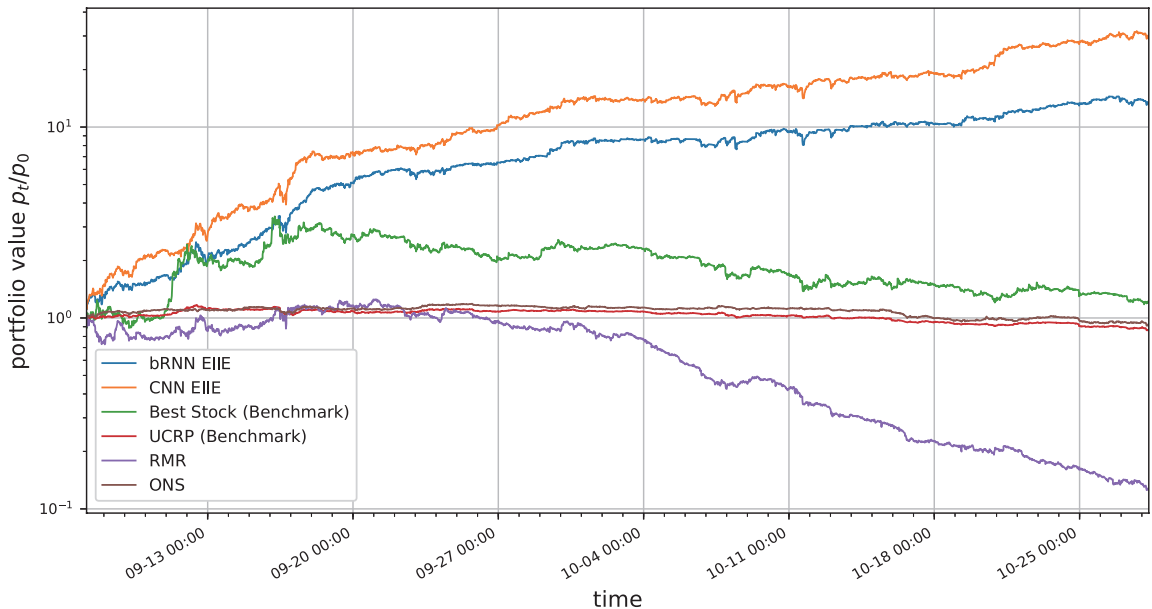}
\caption{Back-Test 1: 2016-09-07-4:00 to 2016-10-28-8:00 (UTC). Accumulated portfolio values (APV, $p_t/p_0$) over the interval of Back-Test 1 for the CNN and
	basic RNN EIIEs, the Best Stock, the UCRP, RMR, and the ONS are plotted in log-10 scale here. The two EIIEs are leading throughout the entire
	time-span, growing consistently only with a few drawdown incidents.
}
\label{fig:backtest1}
\end{figure}

\begin{figure}[!h]
\centering
\includegraphics[width=\linewidth]{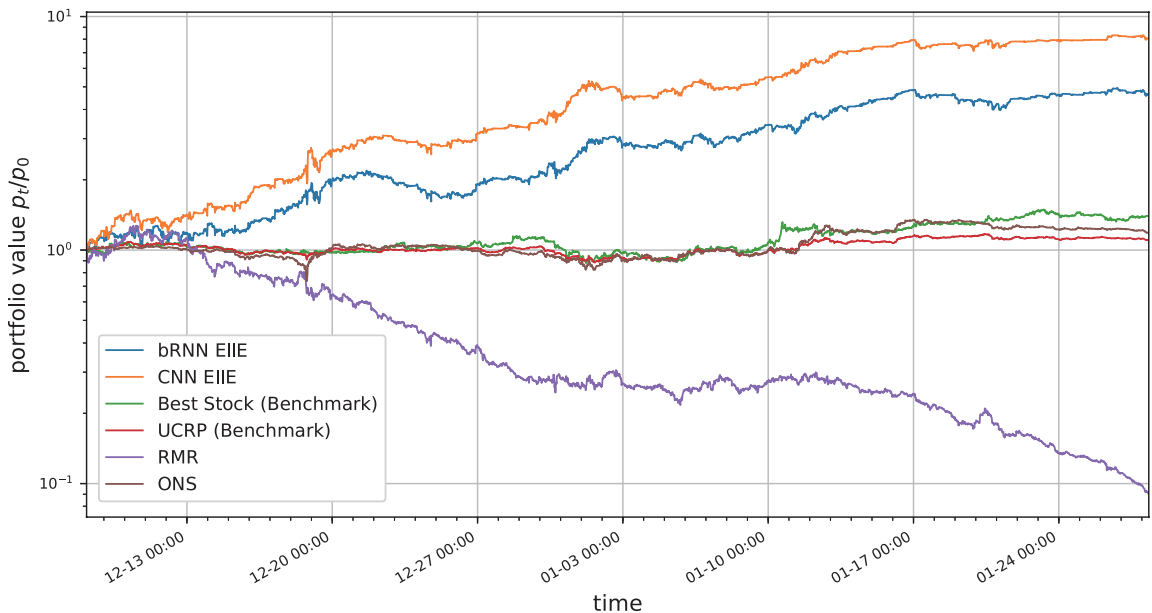}
\caption{Back-Test 2: 2016-12-08-4:00 to 2017-01-28-8:00 (UTC), log-scale accumulated weath. This is the worst experiment amung the three back-tests for the EIIEs.
	However, they are able to steadily climb up till the end of the test.
}
\label{fig:backtest2}
\end{figure}

\begin{figure}[!h]
\centering
\includegraphics[width=\linewidth]{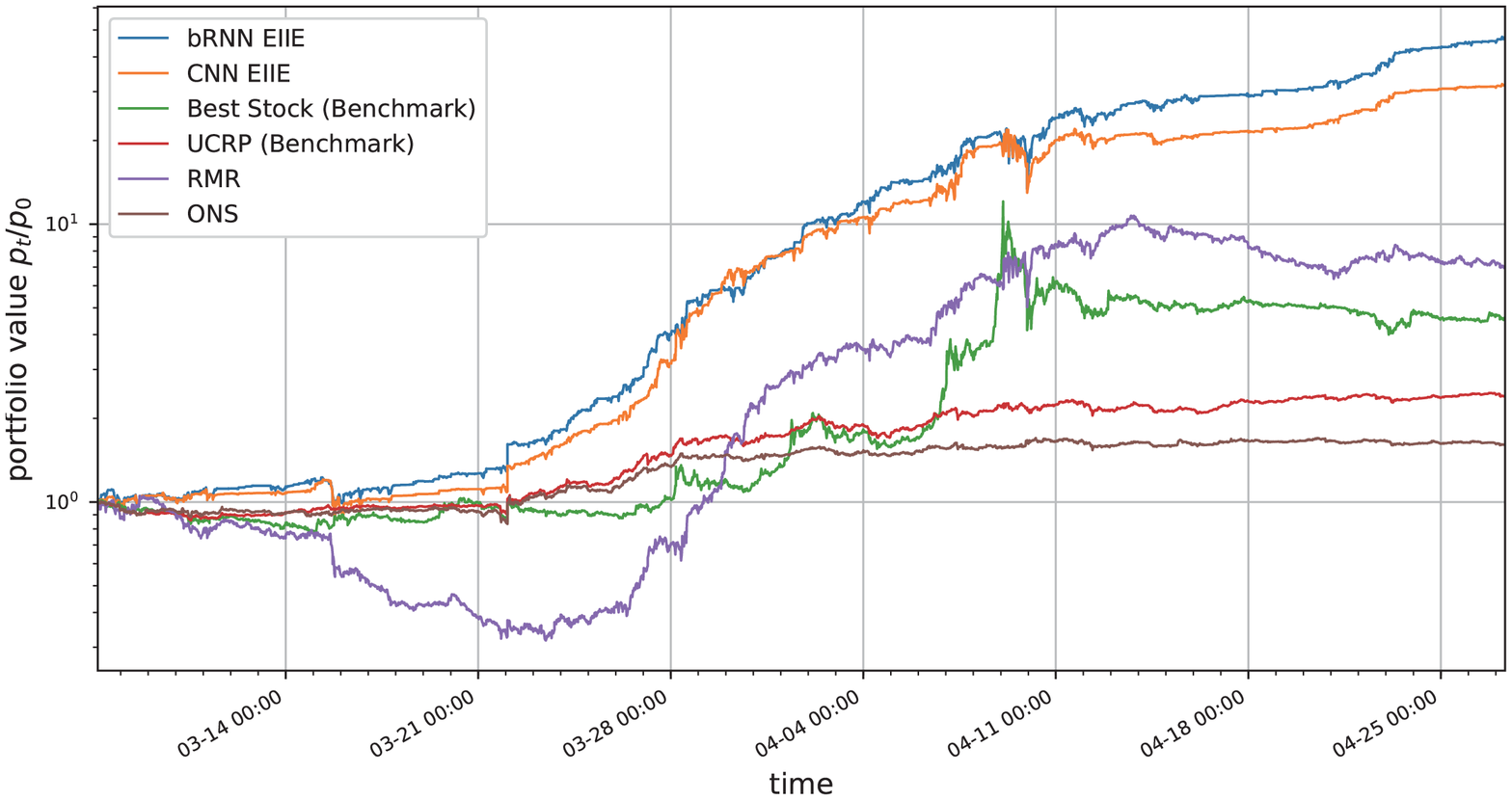}
\caption{Back Test 3: 2017-03-07-4:00 to 2017-04-27-8:00 (UTC), log-scale accumulated weath.
	All algorithms struggle and consolidate at the beginning of this experiment, and
	both of the EIIEs experience two major dips on March 15 and April 9. This diving contributes to their high Maximum Drawdown in the text
	(Table~\ref{tab:performance}). Nevertheless, this is the best month for both EIIEs in term of final wealth.
}
\label{fig:backtest3}
\end{figure}

\section{Conclusion}\label{conclusion}
This article proposed an extensible reinforcement-learning framework solving the general financial portfolio management problem. 
Being invented to cope with multi-channel market inputs and directly output portfolio weights as the market actions,
the framework can be fit in with different deep neural networks, and is linearly scalable with the portfolio size.
This scalability and extensibility are the results of the EIIE meta topology, which is able to accommodate many types of
weight-sharing neural-net structures in the lower level. To take transaction cost into account when training the policy networks,
the framework includes a portfolio-weight memory, the PVM, allowing the portfolio-management agent to learn restraining from
oversized adjustments between consecutive actions, while avoiding the gradient vanishing problem faced by many recurrent networks. 
The PVM also allow parallel training within batching, beating recurrent approaches in learning efficiency to the transaction cost problem.
Moreover, the OSBL scheme governs the online learning process, so that the agent can continuously digest constant incoming market information while trading.
Finally, the agent was trained using a fully exploiting deterministic policy gradient method, aiming to maximize the accumulated wealth as the reinforcement
reward function.

The profitability of the framework surpasses all surveyed traditional portfolio-selection methods, as demonstrated in the paper
by the outcomes of three back-test experiments over different periods in a cryptocurrency market. In these experiments,
the framework was realized using three different underlining networks, a CNN, a basic RNN and a LSTM. All three versions better performed 
in final accumulated portfolio value than other trading algorithms in comparison. The EIIE networks also monopolized the top three positions in the risk-adjusted
score in all three tests, indicating the consistency of the framework in its performances.
Another deep reinforcement learning solution, previously introduced by the authors, was assessed and compared as well under the same settings, losing too to 
the EIIE networks, proving that the EIIE framework is a major improvement over its more primitive cousin.

Among the three EIIE networks, LSTM had much lower scores than the CNN and the basic RNN.
The significant gap in performance between the two RNN species under the same
framework might be an indicator to the well-known secret in financial markets, that history repeats itself. Not being designed to forget its input history, a vanilla RNN
is more able than a LSTM to exploit repetitive patterns in price movement for higher yields.
The gap might also be due to lack of fine-tuning in hyper-parameters for the LSTM. In the experiments, same set of structural hyper-parameters were used for both basic RNN and LSTM.

Despite the success of the EIIE framework in the back-tests, there is room for improvement in future works. The main weakness of the current work 
is the assumptions of zero market impact and zero slippage. In order to consider market impact and slippage, large amount of well-documented real-world
trading examples will be needed as training data. Some protocol will have to be invented for documenting trade actions and market reactions. If that is accomplished,
live trading experiments of the auto-trading agent in its current version can be recorded, for its future version to learn the principles behind
market impacts and slippages from this recorded history. 
Another shortcoming of the work is that the framework has only been tested in one market. To test its adaptability, the current and later versions will need to
be examined in back-tests and live trading in a more traditional financial market.
In addition, the current award function will have to be amended, if not abandoned, for the reinforcement-learning agent
to include awareness of longer-term market reactions. This may be achieved by a critic network. However, the backbone of the current framework, including
the EIIE meta topology, the PVM, and the OSBL scheme, will continue to take important roles in future versions.

\appendix

\section{Proof of Theorem~\ref{thm:convergence}} \label{section:proof}

In order to prove Theorem~\ref{thm:convergence}, it is handy to have the following five lemmas.
\begin{lemma} \label{lem:mono}
	The function $f(\mu)$ in Theorem~\ref{thm:convergence} is monotonically increasing. In other words, $f(\mu_2) \geqslant f(\mu_1)$ if $\mu_2 > \mu_1$.
\end{lemma}
\begin{proof}
	Recall that from Section~\ref{section:transaction},
	\[
	f(\mu) := 
    	    \frac{1}
	     { 1 - c_\mathrm{p} w_{t,0} }
	    \left[
	    1 - c_\mathrm{p} w'_{t,0} - (c_\mathrm{s}+c_\mathrm{p} - c_\mathrm{s}c_\mathrm{p})
	     \sum_{i=1}^m (w'_{t,i} - \mu w_{t,i})^+ 
        \right],
	\]
	The fact that the linear rectifier $(x)^+$ {\Large(}$(x)^+ = x$ if $x>0$, $(x)^+ =0$ otherwise{\Large)} is monotonically increasing readily
	implies that $f(\mu)$ is also monotonically increasing.
\end{proof}

\begin{lemma} \label{lem:lower_bound}
	\[f(0)>0.\]
\end{lemma}
\begin{proof} Using the fact that $w_{t,0},w'_{t,0} \in [0,1]$,
    \begin{align*}
	f(0) &=
    	    \frac{1}
	     { 1 - c_\mathrm{p} w_{t,0} }
	    \left[
	        1 - c_\mathrm{p} w'_{t,0} - (c_\mathrm{s}+c_\mathrm{p} - c_\mathrm{s}c_\mathrm{p})
	           \sum_{i=1}^m (w'_{t,i})^+ 
            \right] \\
	    &=
	    \frac{1}
	     { 1 - c_\mathrm{p} w_{t,0} }
	    \left[
		    1 - c_\mathrm{p} w'_{t,0} - (c_\mathrm{s}+c_\mathrm{p} - c_\mathrm{s}c_\mathrm{p})
		    (1-w'_{t,0})
	    \right] \\
	    &\geqslant 
		    1 - 2 c_\mathrm{p} - c_\mathrm{s} + c_\mathrm{s}c_\mathrm{p} >0,
    \end{align*}
	for $c_\mathrm{p}, c_\mathrm{s} < 38\%$. For a commission rate, $38\%$ is impractically high. Therefore, $f(0)>0$ 
	always holds. 
\end{proof}

\begin{lemma} \label{lem:upper_bound}
	\[f(1)\leqslant 1.\]
\end{lemma}
\begin{proof}
	The proof is split into two cases.
	The fact that $c_\mathrm{p},c_\mathrm{p} \in [0,1)$ implies $(c_\mathrm{s}+c_\mathrm{p} - c_\mathrm{s}c_\mathrm{p}) \geqslant 0$
	will be used in Case 1.
	\begin{description}
		\item[Case 1: $w'_{t,0} \geqslant w_{t,0}.$]
			Since $1 - c_\mathrm{p} w_{t,0} > 0$,
			\begin{align*}
				f(1) &= \frac{1 - c_\mathrm{p} w'_{t,0}}{1 - c_\mathrm{p} w_{t,0}}
					- \frac{1}{1 - c_\mathrm{p} w_{t,0}}
					 (c_\mathrm{s}+c_\mathrm{p} - c_\mathrm{s}c_\mathrm{p})
					  \sum_{i=1}^m (w'_{t,i} - w_{t,i})^+ 
					  \\
				 & \leqslant
				      1 - \frac{1}{1 - c_\mathrm{p} w_{t,0}}
					 (c_\mathrm{s}+c_\mathrm{p} - c_\mathrm{s}c_\mathrm{p})
					  \sum_{i=1}^m (w'_{t,i} - w_{t,i})^+ 
					\\
				 & \leqslant 1
			\end{align*}
		\item[Case 2: $w'_{t,0} < w_{t,0}.$]
			This will be proved by contradiction. By assuming $f(1)>1$,
			\[
		    		1 - c_\mathrm{p} w'_{t,0} - (c_\mathrm{s}+c_\mathrm{p} - c_\mathrm{s}c_\mathrm{p})
				 \sum_{i=1}^m (w'_{t,i} - w_{t,i})^+
				  > 1 - c_\mathrm{p} w_{t,0}.
			\]
			Bringing the two $w_0$'s together,
			\begin{equation} \label{eq:A1}
				c_\mathrm{p}(w_{t,0} - w'_{t,0}) > (c_\mathrm{s}+c_\mathrm{p} - c_\mathrm{s}c_\mathrm{p})
				 \sum_{i=1}^m (w'_{t,i} - w_{t,i})^+.
			\end{equation}
			Noting that
			$
			 w'_{t,0} + \sum\limits_{i=1}^m w'_{t,i} = 1 = w_{t,0} + \sum\limits_{i=1}^m w_{t,i}
			$,
			\[
				w_{t,0} - w'_{t,0} = 1 - \sum_{i=1}^m w_{t,i} - \left(1 - \sum_{i=1}^m w'_{t,i}\right)
				= \sum_{i=1}^m \left(w'_{t,i} - w_{t,i}\right).
			\]
			Using identity $(a-b)^+ - (b-a)^+ = a-b$, \eqref{eq:A1} becomes
			\[
				c_\mathrm{p}  \left[ \sum_{i=1}^m (w'_{t,i} - w_{t,i})^+ 
				- \sum_{i=1}^m (w_{t,i} - w'_{t,i})^+ \right]
				> (c_\mathrm{s}+c_\mathrm{p} - c_\mathrm{s}c_\mathrm{p})
				 \sum_{i=1}^m (w'_{t,i} - w_{t,i})^+.
			\]
			Moving the $(w'_{t,i} - w_{t,i})^+$ terms to the right-hand side,
			\begin{equation} \label{eq:contradition}
				-c_\mathrm{p} \sum_{i=1}^m (w_{t,i} - w'_{t,i})^+ >
				c_\mathrm{s}(1-c_\mathrm{p}) \sum_{i=1}^m (w'_{t,i} - w_{t,i})^+.
			\end{equation}
			The left-hand side of \eqref{eq:contradition} is a non-positive number, and the right-hand side is a non-negative number.
			The former is greater than the latter, arriving at a contradiction.
	\end{description}
	Therefore, $f(1)\leqslant 1$ in both cases.
\end{proof}

\begin{lemma} \label{lem:mu0}
	the sequence $\left\{\tilde{\mu}^{(k)} _t\right\}$, defined as
	\begin{equation*}
		\left\{\tilde{\mu}^{(k)} _t \left| \tilde{\mu}^{(0)} _t = 0 \;\mathrm{and}\; 
			\tilde{\mu}^{(k)} _t = f\left(\tilde{\mu}^{(k-1)} _t\right),\;
		k \in \mathbb{N}_0 \right. \right\}
	\label{eq:mu0_sequence}
	\end{equation*}
	converges to $\mu_{t}$.
\end{lemma}
\begin{proof}
	This is a special case of the final goal Theorem~\ref{thm:convergence} when $\mu_\odot = 0$.
	This convergence is proved by the Monotone Convergence Theorem (MCT) \cite[Chapter~5]{rudin1964principles}.
	The monotonicity of $f$ by Lemma~\ref{lem:mono} with Mathematical Induction establishes 
	an upper bound for $\left\{\tilde{\mu}^{(k)} _t\right\}$. 
	\begin{align*}
	\left.
	\begin{array}{r@{\mskip\thickmuskip}l}
	    \tilde{\mu}^{(0)} _t = 0 & < \mu_{t}, \\
	    \mathrm{If}\;
	    \tilde{\mu}^{(k-1)} _t \leqslant \mu_{t},\; 
	         \tilde{\mu}^{(k)} _t = f\left( \tilde{\mu}^{(k-1)} _t \right) & \leqslant f(\mu_{t}) = \mu_{t}
        \end{array}
        \right\}
	\implies
        \tilde{\mu}^{(k)} _t \leqslant \mu_{t},\; \forall k.
	\end{align*}
	Note that by definition $\mu_{t}$ is the transaction remainder factor, and $0<\mu_t\leqslant 1$.
	The monotonicity of sequence $\left\{\tilde{\mu}^{(k)} _t\right\}$ itself can be also proved by Mathematical Induction and
	Lemma~\ref{lem:lower_bound}.
	\begin{align*}
	\left.
	\begin{array}{r@{\mskip\thickmuskip}l}
		\tilde{\mu}^{(1)} _t = f\left( 0 \right) & >
	    0 = \tilde{\mu}^{(0)} _t, \\
	    \mathrm{If}\;
	    \tilde{\mu}^{(k-1)} _t \geqslant \tilde{\mu}^{(k-2)} _t,\;
	    \tilde{\mu}^{(k)} _t = f\left( \tilde{\mu}^{(k-1)} _t \right) &\geqslant f\left( \tilde{\mu}^{(k-2)} _t \right)
	        = \tilde{\mu}^{(k-1)} _t
        \end{array}
        \right\}
	\implies
        \tilde{\mu}^{(k)} _t \geqslant \tilde{\mu}^{(k-1)} _t ,\; \forall k.
	\end{align*}
	If $\tilde{\mu}^{(k)} _t = \tilde{\mu}^{(k-1)} _t$, then $\tilde{\mu}^{(k)} _t$ is the solution of Equation~\eqref{eq:mu_iteration},
	and the proof ends here. Otherwise, the sequence $\left\{\tilde{\mu}^{(k)} _t\right\}$ is strictly increasing and bounded above
	by $\mu_{t}$. In that case, by MCT, $\lim\limits_{k\to\infty} \tilde{\mu}^{(k)} _t =\mu^*$,
	where $\mu^*$ is the Least Upper Bound of $\left\{\tilde{\mu}^{(k)} _t\right\}$.
	As a result,
	$
	    0 = \lim\limits_{k\to\infty} \left(\tilde{\mu}^{(k+1)} _t - \tilde{\mu}^{(k)} _t \right) = 
	        \lim\limits_{k\to\infty} \left(
	            f\left(\tilde{\mu}^{(k)} _t\right) - \tilde{\mu}^{(k-1)} _t
		\right)
	      = f\left( \mu^* \right) - \mu^*.
	$
	Therefore, $\mu^*$ is the solution to Equation~\eqref{eq:mu_iteration}, and hence 
	$\lim\limits_{k\to\infty} \tilde{\mu}^{(k)} _t =\mu_{t}$.
\end{proof}

\begin{lemma} \label{lem:mu1}
	the sequence $\left\{\tilde{\mu}^{(k)} _t\right\}$, defined as
	\begin{equation*}
		\left\{\tilde{\mu}^{(k)} _t \left| \tilde{\mu}^{(0)} _t = 1 \;\mathrm{and}\; 
			\tilde{\mu}^{(k)} _t = f\left(\tilde{\mu}^{(k-1)} _t\right),\;
		k \in \mathbb{N}_0 \right. \right\}
	\label{eq:mu1_sequence}
	\end{equation*}
	converges to $\mu_{t}$.
\end{lemma}
\begin{proof}
	The proof is similar to that of Lemma~\ref{lem:mu0} using Mathematical Induction and MCT.
	The sequence is monotonically decreasing and bounded below by $\mu_t$.
	The monotonically is a result of Lemma~\ref{lem:upper_bound}, and the boundedness is by Lemma~\ref{lem:mono}.
\end{proof}

With the previous two Lemmas, it is in a good position to prove the general convergence theorem.
Recall Theorem~\ref{thm:convergence} from Section~\ref{section:transaction}:
\convergence*
\begin{proof}
    It is proved in three cases:
    \begin{description}	
	\item[Case 1: $\mu_\odot = \mu_t,\;0,\;\mathrm{or}\;1$.] This case is trivial, as when $\mu_\odot = \mu_t$ it is the solution of $\mu=f(\mu)$,
		and the sequence will obviously converge. The convergence to $\mu_t$ for the other two values of $\mu_\odot$ is guaranteed by 
		Lemma~\ref{lem:mu0} and \ref{lem:mu1}
	\item[Case 2: $0 < \mu_\odot < \mu_t$.]
		Constructing a sequence $\left\{\widehat{\mu}^{(k)} _t\right\}$ using $\mu_\odot = 0$,
		\[
		\left\{\widehat{\mu}^{(k)} _t \left| \widehat{\mu}^{(0)} _t = 0 \;\mathrm{and}\; 
			\widehat{\mu}^{(k)} _t = f\left(\widehat{\mu}^{(k-1)} _t\right),\;
		k \in \mathbb{N}_0 \right. \right\}.
		\]
		By the proof of Lemma~\ref{lem:mu0}, $\left\{\widehat{\mu}^{(k)} _t\right\}$ is strictly increasing and bounded above by $\mu_t$,
		so there is a $j\in \mathbb{N}_0$ such that
		\[
			\widehat{\mu}^{(j)} \leqslant \mu_\odot \leqslant \widehat{\mu}^{(j+1)}.
		\]
		If any of the above equality holds, the two sequences coincide from $j+1$ onward, converging to $\mu_t$, and the proof ends here. Otherwise,
		\[
			\widehat{\mu}^{(j)} < \mu_\odot < \widehat{\mu}^{(j+1)}.
		\]
		Using the monotonicity of $f(\mu)$ by Lemma~\ref{lem:mono}, these inequalities become
		\[
			\widehat{\mu}^{(j+1)} = f(\widehat{\mu}^{(j)}) \leqslant 
			\tilde{\mu}^{(1)} = \mu_\odot \leqslant
			f(\widehat{\mu}^{(j+1)}) = \widehat{\mu}^{(j+2)}.
		\]
		Again, if one of the equality holds, the proof ends here. This chain can go on indefinitely if no equality holds,
		\[
			\widehat{\mu}^{(j+k+1)} < 
			\tilde{\mu}^{(k)}  <
			 \widehat{\mu}^{(j+k+2)}.
		\]
		By the Squeeze Theorem \cite{leithold1996calculus},
		\[
			\lim_{k\to \infty} \tilde{\mu}^{(k)} = \lim_{k\to \infty} \widehat{\mu}^{(k)}  = \mu_t.
		\]
	 \item[Case 3: $1 > \mu_\odot > \mu_t$.]
		 This case is proved in a similar way to Case 2, by constructing a sequence with $\mu_\odot=1$
		 and making use of Lemma~\ref{lem:mu1}.
    \end{description}
    In conclusion, for any initial value $\mu_\odot \in [0,1]$, the sequence $\tilde{\mu}^{(k)}$ converges to $\mu_t$.
\end{proof}

\section{Hyper-Parameters}
\newcolumntype{L}[1]{>{\raggedright\let\newline\\\arraybackslash\hspace{0pt}}m{#1}}
\begin{table}[!htb]
	\centering
	\begin{tabular}{L{8em}|r|L{24em}} 
		hyper-parameters & value & description\\ 
		\hlineB{3}
		batch size & 50 & Size of mini-batch during training. (Section~\ref{section:preselection})\\
		\hline
		window size & 50 & Number of the columns (number of the trading periods) in each input price matrix. (Section~\ref{section:price_tensor})  \\
		\hline
		number of assets & 12 & Total number of preselected assets (including the cash, Bitcoin). (Section~\ref{section:preselection}) \\ 
		\hline
		trading period (second)& 1800 & Time interval between two portfolio redistributions. (Section~\ref{section:period})\\
		\hline
		total steps & $2\times 10^6$ & Total number of steps for pre-training in the training set.\\
		\hline
		regularization coefficient & $10^{-8}$ &  The L2 regularization coefficient applied to network training.\\
		\hline
		learning rate & $3 \times 10^{-5}$ & Parameter $\alpha$ (i.e. the step size) of the Adam optimization \cite{Kingma2014}.\\
		\hline
		volume observation (day) & 30 & The length of time during which trading volumes are used to preselect the portfolio assets.
		(Section~\ref{section:preselection})\\
		\hline
		commission rate & $0.25\%$ & Rate of commission fee applied to each transaction. (Section~\ref{section:transaction})\\
		\hline
		rolling steps & 30 & Number of online training steps for each period during cross-validation and back-tests. \\
		\hline
		sample bias & $5\times 10^{-5}$ & Parameter of geometric distribution when selecting online training sample batches.
		(The $\beta$ in Equation~\ref{eq:prob_batch} of Section~\ref{section:training})\\
	\end{tabular}
	\caption{Hyper-parameters of the reinforcement-learning framework. They are chosen based on the networks' scores in the cross-validation
		set described in Table~\ref{tab:datarange} of Section~\ref{section:ranges}.
		Although these are the values used in the experiments of the paper, they are
		all adjustable in the framework.
	}
	\label{tab:hyperpar}
\end{table}
The hyper-parameters and their values used in the back-test experiments of the paper are listed in Table~\ref{tab:hyperpar}. These numbers are
selected to maximize the network scores in the cross-validation time-range (see Section~\ref{section:ranges}).
In order to avoid over-fitting, the cross-validation range and the back-tests do not overlap.

Different topologies of the IIEs were also tried in the cross-validation set, and it turned out that deeper network structures than those
presented in Figure~\ref{fig:cnn_topology} and \ref{fig:rnn_topology} did not improve scores on the set.


\bibliography{reference}

\end{document}